\documentclass[lettersize,journal]{IEEEtran}
\usepackage{amsmath,amsfonts}
\usepackage{amssymb}
\usepackage{bbm}
\usepackage{amsthm}
\usepackage{algorithmicx}
\usepackage{algorithm}
\usepackage{array}
\usepackage{color}
\usepackage[caption=false,font=normalsize,labelfont=sf,textfont=sf]{subfig}
\usepackage{makecell}
\usepackage{textcomp}
\usepackage{tabularx}
\usepackage{stfloats}
\usepackage{url}
\usepackage{verbatim}
\usepackage{graphicx}
\usepackage{cite}
\usepackage{multirow}
\usepackage{multicol}
\usepackage{ragged2e}
\newtheorem{Lemma}{ \textit{Lemma}}

\newtheorem{pro}{Proposition}
\newtheorem{thm}{Theorem}
\newtheorem{Col}{Corollary}
\usepackage{algpseudocode}

\usepackage{CJKutf8}
\begin{document}

\title{Stochastic Geometry Analysis and Optimization for Multi-LEO Satellite-Assisted Communications}

\author{Xu Hu,~\IEEEmembership{Graduate Student Member,~IEEE}, Bin Lin,~\IEEEmembership{Senior Member,~IEEE}, Ping Wang,~\IEEEmembership{Fellow,~IEEE}, Xiao Lu, ~\IEEEmembership{Member,~IEEE}, Weihua Zhuang,~\IEEEmembership{Fellow,~IEEE}, and Shuang Qi
\vspace{-1\baselineskip}
\thanks{
The work was supported in part by the National Natural Science Foundation of China (No. 62371085), in part by Dalian Science and Technology Innovation Fund (No. 2025RG06), and in part by Guangxi Science and Technology Infrastructure and Talent Development Program (No. GUIKE AD25069109). An earlier version of this paper was accepted in part at the IEEE International Conference on Communications (ICC 2026), Glasgow, Scotland, UK, 2026. {\it{(Corresponding author: Bin Lin.)}}

Copyright (c) 20xx IEEE. Personal use of this material is permitted. However, permission to use this material for any other purposes must be obtained from the IEEE by sending a request to pubs-permissions@ieee.org. Xu Hu and Bin Lin are with the Information Science and Technology College, Dalian Maritime University, Dalian 116026, China (e-mail: huxu@dlmu.edu.cn; binlin@dlmu.edu.cn). Ping Wang is with the Department of Electrical Engineering and Computer Science, York University, Toronto, ON M3J 1P3, Canada (e-mail: pingw@yorku.ca). Xiao Lu is with the Research and Development, Ericsson, Ottawa, ON K2K 2V6, Canada (e-mail: xiao.lu@ericsson.com). Weihua Zhuang is with the Department of Electrical and Computer Engineering, University of Waterloo, Waterloo, ON N2L 3G1, Canada (e-mail: wzhuang@uwaterloo.ca). Shuang Qi is with the School of Electronic and Information Engineering, Liaoning University of Technology, Jinzhou 121001, China (e-mail: dxqishuang@lnut.edu.cn)

This manuscript has been accepted by IEEE Transactions on Vehicular Technology, DOI: 10.1109/TVT.2026.3690149.
}
}



\maketitle

\begin{abstract}

Low Earth orbit (LEO) satellites are vital in providing seamless connectivities to ground users (GUs). However, complex satellite-to-ground propagation and non-negligible satellite failure probabilities make single-satellite transmission unreliable. To address this issue, we consider a multi-LEO-satellite communication network (mLEO-CN), where a typical GU is served by multiple LEO satellites. The satellite locations are modeled as a binomial point process on a sphere and ordered according to their distances to the GU. To improve communication reliability, a serial communication scheme is adopted, in which the subsequent satellite is activated only when all preceding satellites fail to satisfy the target signal-to-noise ratio threshold. By employing the tool of stochastic geometry, we derive key performance metrics, including overall transmission success probability, average transmission rate, and end-to-end delay under the serial communication scheme by incorporating both shadowed Rician (SR) fading and satellite failure probabilities. To improve analytical tractability, we derive asymptotic expressions for the end-to-end transmission success probability in two special cases by leveraging the asymptotic expression of the SR fading model and a Gamma-distribution, and further propose a cluster-based approximation approach for the overall transmission success probability. Building on these analyses, we formulate an optimization problem to maximize the ratio of the overall transmission success probability to the end-to-end delay, and obtain its solution using a one-dimensional search approach. Numerical results validate the analysis, reveal the impact of key parameters, and identify the optimal number of serving LEO satellites under different constellation altitudes.

\end{abstract}


\section{Introduction}
\IEEEPARstart{I}{n} recent years, low Earth orbit (LEO) satellites have risen as a promising solution for providing seamless network services in underserved areas, including deserts, oceans, and rural regions, thereby facilitating ubiquitous and high-capacity connectivities \cite{LEO1}. Currently, many LEO satellite constellation projects have launched a series of LEO satellites, driving various new services such as direct user-to-satellite communications and inter-satellite communications \cite{project1}. Despite these advancements, the complex and dynamic nature of the satellite-to-ground (as well as ground-to-satellite) propagation environment poses severe obstacles to maintaining high-quality links. Specifically, due to poor communication conditions, such as severe channel fading induced by signal scattering, these limitations make it difficult to meet traffic demands from ground users (GUs). Meanwhile, given the large scale of the constellation and the long operational lifespan of LEO satellites, they are increasingly vulnerable to threats stemming from various factors, such as high-energy particles, solar storms, component malfunctions, and anti-satellite weapons, resulting in temporarily disabling or permanently destroying satellites \cite{failure3}. Consequently, communication reliability with LEO satellites decreases over time. Specifically, the reliability factor, initially at one, gradually decreases to 0.8 within two years and further drops to 0.6 within five years \cite{decade}. Therefore, reliance on a single satellite is highly vulnerable, as it suffers not only from unstable link quality but also from the potential for service interruption.

In light of these threats, constructing a multi-LEO satellite-assisted communication network (mLEO-CN), where a GU is served by multiple LEO satellites, is essential to enhance the reliability of LEO communication networks. This network leverages the enhanced signal strength from multiple LEO satellites, mitigating the limitations of poor signal reception from an individual satellite \cite{Zhang24Multi}. To quantitatively evaluate the performance benefits of multiple LEO satellites, it is vital to establish a theoretical framework for assessing the performance of mLEO-CN in comparison with the individual LEO satellite-served network.

Stochastic geometry is a potent mathematical instrument extensively utilized to characterize the system-level performance of large-scale wireless networks \cite{Lu3}. By adopting distribution models such as the binomial point process (BPP) and the Poisson point process (PPP) to represent the locations of network nodes, this methodology has provided valuable insights into critical performance metrics across various network domains. Recent applications of stochastic geometry in analyzing large-scale LEO satellite communication systems have underscored its utility in evaluating the performance of LEO satellite networks \cite{okati2020downlink,22-9,23-1}. Building on this foundation, our objective is to construct an analytical framework, leveraging stochastic geometry, to characterize the network performance of the mLEO-CN and to delineate the effects of key network parameters on the end-to-end signal transmission performance.  

\subsection{Related Works}
Most existing studies concentrate on assessing network performance facilitated by an individual satellite, with particular emphasis on the nearest LEO satellite serving a GU. In the early stage, Yastrebova \emph{et al.} analyzed the uplink coverage probability (CP) for a terrestrial–satellite network by modeling the distribution of LEO satellites as the uniform distribution along each orbit \cite{Nearest-D-U1}. However, the computational complexity increases with the rising number of orbits in this model. Subsequently, the accuracy of modeling LEO satellite constellations as a BPP was validated by establishing a connection between the BPP and the Fibonacci lattice \cite{okati2020downlink}. Furthermore, the contact distance distribution for the BPP on spherical surfaces was proposed in \cite{talgat2020nearest}. Building on these foundational studies, a growing number of studies have since explored the performance of LEO satellite networks using BPP models. In \cite{okati2020downlink}, Okati \emph{et al.} constructed frameworks to capture the downlink CP and average data rate of a single-layer LEO satellite network by employing the Rayleigh fading model. Then, Talgat \emph{et al.} extended the single-layer LEO satellite network to a multi-layer spherical-shell architecture and conducted a corresponding coverage probability analysis \cite{Talgat2021Stochastic}. To accurately describe the characteristics of the satellite-to-ground link, the precise outage probability was derived under the shadowed Rician (SR) fading model \cite{onlynoise4}. Since then, the performance analysis of modeling LEO satellites as BPPs and using the SR fading model has been explored gradually.

Since the received signal strength in a BPP model depends on each receiver's location, the network performance of LEO satellite communication networks is influenced by each LEO satellite. Consequently, computational complexity increases with the number of LEO satellites, making performance analysis more challenging with a BPP model than with a homogeneous spherical PPPs (SPPPs) model \cite{23-1}. To this end, the SPPP model has been applied to LEO satellite networks for tractable analysis. For instance, the downlink CP under the Nakagami-$m$ fading model was derived in \cite{23-1} and the moments and meta-distribution of conditional CP were conducted in \cite{sun2023fine}. Despite a series of studies employing the SPPP to model the distribution of LEO satellites, the suitability of the SPPP model for LEO networks with finite areas and limited nodes remains controversial.


Currently, many complex distribution models have been introduced to accurately capture the practical characteristics of LEO satellite communication networks. Specifically, to account for the motion of LEO satellites in their orbital planes into the modeling of delay-tolerant data harvesting architectures, Choi utilized the Cox point process to represent the distribution of LEO satellites \cite{choi2024modeling}. To accurately characterize the safety distance constraint for LEO satellites, Zhang \emph{et al.} in \cite{zhang2025leo} proposed a novel Poisson arc hardcore point process (PAHPP) model to precisely delineate the distribution of LEO satellites operating in circular orbits.

More recently, the research focus has gradually shifted from single-satellite services to communication networks enabled by multiple LEO satellites. Under coherent joint transmissions, Shang \emph{et al.} derived the joint distance distribution between cooperative LEO satellites and a GU, thereby establishing a theoretical foundation for analyzing downlink coverage probability in coordinated transmission scenarios \cite{Shang2023coverage}. Building on this foundation, subsequent works extended the analysis to both downlink and uplink coverage performance \cite{Li2024An, Emara2025A}. After that, Li \emph{et al.} further extended this line of research to cell-free massive multiple-input multiple-output (MIMO) satellite communication networks \cite{Li2026Downlink}. On this basis, the concept of satellite cluster networks was introduced in \cite{Lee2025Analyzing}, followed by the development of a mathematical framework for evaluating coverage probability under both coherent and non-coherent joint transmissions \cite{Lee2025Analyzing, Yang2026Adaptive}. Beyond communication performance, satellite clusters have also been applied to positioning services. For instance, by integrating multiple LEO satellites into existing medium Earth orbit (MEO)-based positioning systems, the $K$-availability and $K$-localizability probabilities were analyzed to quantify localization performance improvements \cite{wang2025analyzing}. Moreover, according to their functional roles within a cluster, LEO satellites can be categorized into leaders and followers, and a corresponding analytical framework was proposed to evaluate the communication performance of leader-follower satellite architectures \cite{Ru2026Coverage}. For parallel transmissions, Qu \emph{et al.} derived analytical expressions for the coverage probability and optimized satellite density to maximize this metric \cite{Qu2026Coverage}. In addition, cooperative transmission based on coordinated beamforming among the several nearest satellites has been studied for further communication reliability enhancement. \cite{kim2024coverage}.

Despite significant progress in multi-satellite-assisted communications, most existing schemes, such as coherent joint transmission, non-coherent joint transmission, and parallel transmission, rely on the simultaneous participation of multiple LEO satellites, resulting in high energy consumption, considerable signaling overhead, and stringent synchronization and coordination requirements. In addition, satellite cluster networks generally require cluster formation, inter-satellite cooperation, and role assignment, which further increase implementation complexity in dynamic LEO environments. Motivated by these limitations, it is of great interest to develop a more efficient communication scheme for the mLEO-CN. In this context, the serial communication scheme is particularly promising, as it activates LEO satellites sequentially, from the nearest to the farthest, such that subsequent satellites are involved only when necessary. This scheme can improve resource utilization, reduce unnecessary energy costs, and enhance communication reliability with low implementation complexity. Nevertheless, the performance characterization of such mLEO-CNs remains challenging. First, the sequential involvement of multiple LEO satellites results in increasing complexity as the number of LEO satellites grows. Second, realistic satellite-to-ground (as well as ground-to-satellite) channels require the adoption of the SR fading model to capture scattering, reflection, diffraction, shadowing, and multipath effects, which further complicates the analysis. Third, the end-to-end performance depends on both uplink and downlink transmissions, making the characterization more difficult. Finally, existing studies rarely account for satellite failures caused by harsh space environments, despite their non-negligible impact on network reliability and performance. This drives our research to tackle the aforementioned challenges and bridge the existing gap.

\subsection{Contributions and Organization}

In this paper, we consider an mLEO-CN, where a GU is served by multiple LEO satellites. To enhance communication reliability, we adopt a serial communication scheme in which LEO satellites provide network services to the GU sequentially, from the closest to the farthest, and each subsequent satellite is involved only if all preceding satellites fail to satisfy the predefined SNR threshold. The main technical contribution of this paper is the development of a tractable analytical framework for performance characterization under the serial communication scheme. On this basis, the asymptotic analysis, the cluster-based approximation approach, and optimization are complementary developments that enhance theoretical understanding, improve analytical tractability, and facilitate practical design of multi-LEO-satellite-assisted communication systems. To better reflect practical operating conditions, we incorporate failure probability into the stochastic geometry analysis of LEO satellite networks to capture satellite malfunctions in space environments. The contributions of this paper are summarized as follows.


\begin{itemize}
\item We develop a tractable analytical framework for the mLEO-CN to characterize the end-to-end transmission performance under the serial communication scheme. Specifically, by leveraging stochastic geometry and jointly accounting for channel fading and potential satellite failures, the proposed framework enables the evaluation of key performance metrics, including the overall transmission success probability, average transmission rate capacity, and end-to-end delay. To further improve analytical tractability and provide additional theoretical insight, we also derive asymptotic closed-form expressions for the end-to-end transmission success probability in two special cases.


\item Building upon the developed analytical framework, we further propose a cluster-based approximation approach to improve the tractability of evaluating the overall transmission success probability. In this approach, the serving LEO satellites are partitioned into multiple clusters, and the performance of each cluster is represented by that of its central satellite. Accordingly, the cumulative product associated with all serving satellites is approximated by the corresponding power of the representative satellite’s performance within each cluster, thereby significantly reducing computational complexity while maintaining high accuracy.

\item Extensive simulations are conducted to validate the proposed analytical framework and verify the accuracy of the derived exact and asymptotic expressions. In addition, the simulation results confirm the effectiveness of the proposed cluster-based approximation approach and illustrate the impact of key parameters, including constellation size, constellation altitude, predefined SNR threshold, and satellite failure probability, on the end-to-end transmission performance.

\item Based on the proposed analytical framework, we formulate an optimization problem to maximize the ratio of the overall transmission success probability to the end-to-end delay. The corresponding optimization results identify the optimal number of serving satellites at different constellation altitudes, thereby providing practical design guidance for LEO satellite networks. 
  \end{itemize}

The rest of this paper is organized as follows. Section \uppercase\expandafter{\romannumeral2} describes the system model, including the network model, communication scheme, channel model, and signal model. Section \uppercase\expandafter{\romannumeral3} derives the key transmission performance metrics for the mLEO-CN, including the overall transmission success probability, the average transmission rate capacity, and the end-to-end delay. Moreover, this section provides asymptotic expressions and an approximation approach for evaluating the overall transmission success probability. Section \uppercase\expandafter{\romannumeral4} shows extensive numerical results and the solutions to the formulated optimization problem. Section \uppercase\expandafter{\romannumeral5} concludes this work.

\section{System Model}
\subsection{Network Model}
 \begin{figure}[t]\label{figure1}
  \centerline{\includegraphics[scale=0.55]{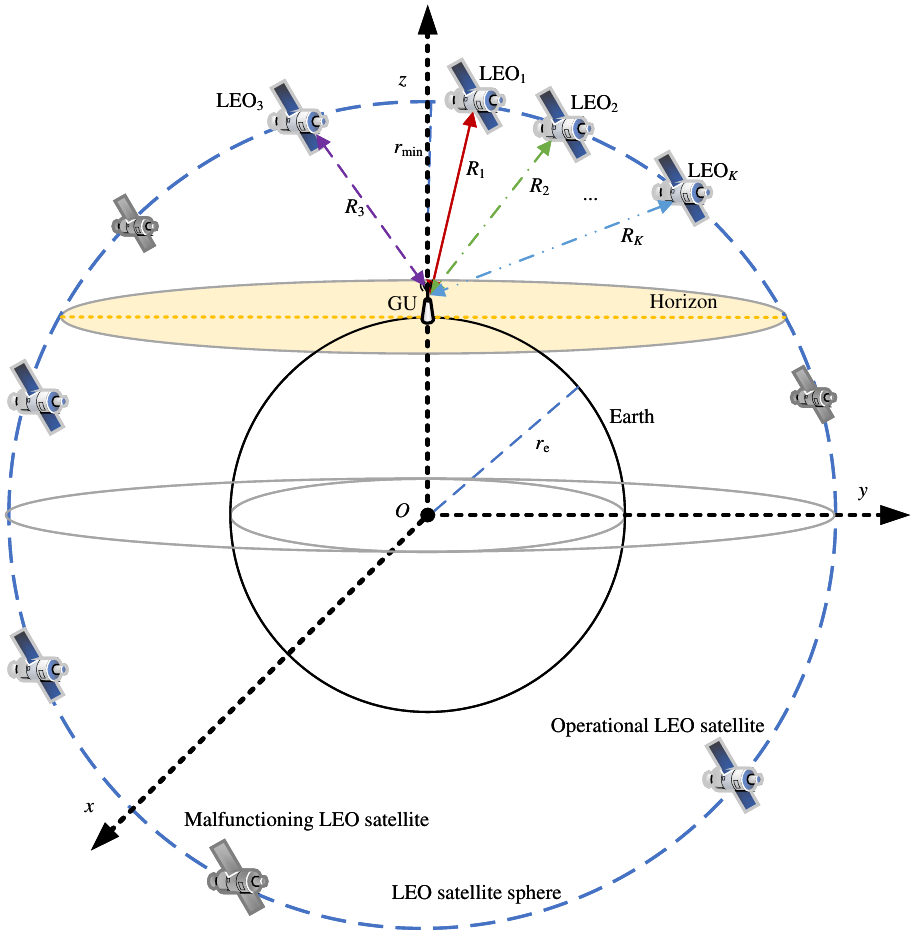}}
 	\caption{A geometric structure of mLEO-CN.}
    \label{figure1}
 \end{figure}
We consider a multi-LEO-satellite communication network (mLEO-CN) that includes an LEO satellite constellation and a typical GU, as shown in Fig. \ref{figure1}. We establish a three-dimensional space using the Cartesian coordinate system, where the Earth's center is at the origin ${O} \buildrel \Delta \over = \left( {0,0,0} \right) \in {{\mathbb{R}}^3}$, the $x$-axis points towards a specific point on the equator, the $y$-axis is perpendicular to the $x$-axis and lies within the equatorial plane, and the $z$-axis points towards the geographic North Pole, perpendicular to the equatorial plane. Without loss of generality, the GU is located on the Earth's surface at observation point $\left( {0,0,{r_{\rm{e}}}} \right)$ \cite{23-1} and can communicate directly with LEO satellites. In the LEO satellite constellation, there are $N$ LEO satellites uniformly distributed on the sphere surface $\mathbb{S} _{{r_{\rm{a}}}}^2$ forming a BPP with radius of ${r_{\rm{a}}}$, where ${r_{\rm{a}}} = {r_{\rm{e}}} + {r_{\min }}$ and ${r_{\min }}$ is the constellation altitude and ${\Phi _{{\rm{LEO}}}}{\rm{ = }}\left\{ {{s_1},{s_2},...,{s_N}} \right\}\left( {N \ge 1,N \in {\mathbb{N}^ + }} \right)$ presents the set of LEO satellites. We rank LEO satellites in the mLEO-CN based on their distance to the GU, designating each of them as the ${k^{{\rm{th}}}}\left( {1 \le k \le N} \right)$ LEO satellite, from the closest to the farthest. Correspondingly, the distance between the ${k^{{\rm{th}}}}$ LEO satellite and the GU is denoted as ${R_k}$. The GU can obtain the ordering of LEO satellites using satellite orbital data and ephemeris information in practical LEO satellite communication systems. Moreover, the maximum possible distance between the GU and an LEO satellite over its horizon, known as the visible distance, is represented by ${r_{\max }}{\rm{ = }}\sqrt {2{r_{\rm{e}}}{r_{\min }} + r_{\min }^2}$ \cite{okati2020downlink}. Accordingly, the LEO satellites over the horizon are referred to as visible satellites, forming a subset of LEO satellites, denoted as ${\Phi _{\rm{vis}}}{\rm{ = }}\left\{ {{s_1},{s_2},...,{s_M}} \right\}\left( {M \le N} \right)$ \cite{CDFofSR2}. In the mLEO-CN involving $K$ LEO satellites, the serving satellites, denoted as ${\Phi _s}{\rm{ = }}\left\{ {{s_1},{s_2},...,{s_K}} \right\}\left( {K \le M} \right)$, are defined as a subset of $K$ visible LEO satellites that are capable of communicating with the GU.

Signal transmissions in the mLEO-CN involve data transfer from the GU to an LEO satellite, followed by processing and subsequent feedback to the GU. To support this, $K$ orthogonal frequency channels are allocated within the Ka-band spectrum, and a decode-and-forward (DF) protocol is employed to facilitate the end-to-end signal transmissions \cite{DFprotocol}. To better characterize the operating state of an LEO satellite, we introduce a satellite failure probability $q_k$ to quantify the likelihood of the ${k^{{\rm{th}}}}$ LEO satellite malfunctions caused by potential threats such as solar storms, high-velocity collisions, and other space hazards. Accordingly, a malfunctioning satellite is treated as unavailable for transmission and is excluded from the set of serving satellites, while the communication scheme remains unchanged for the operational satellites.

\subsection{Communication Scheme}
To enhance the communication reliability of LEO satellite communication networks, we implement a serial communication scheme with the serial probing mechanism in the mLEO-CN to regulate end-to-end signal transmissions. Specifically, the GU initially sends a pilot signal to the nearest LEO satellite over the first channel, where the satellite measures the received SNR. If the measured SNR exceeds a predefined threshold, the satellite responds by sending a feedback pilot signal to the GU on the same channel. The GU then evaluates the SNR of the returned signal. If the SNR at the GU also exceeds the predefined threshold, the GU selects the first LEO satellite for end-to-end signal transmissions. However, if the received SNR at either the first LEO satellite or the GU does not satisfy the threshold, the GU proceeds by sending a pilot signal to the second LEO satellite, and the same measurement and evaluation procedure is repeated to determine if the second satellite can be selected. This iterative procedure continues until a suitable LEO satellite is identified for transmission. If none of the $K$ satellites in the mLEO-CN can satisfy the predefined SNR threshold requirement, the signal transmission process is terminated.

\subsection{Channel Model} 
\subsubsection{Path-loss Model} 
A distance-dependent power path loss model is utilized to represent the path loss between the GU and the ${k^{{\rm{th}}}}$ LEO satellite \cite{24-47}, that is, $L_k\left( {{R_k}} \right) = {\left( {\frac{c }{{4\pi {f_c}}}} \right)^2}R_k^{ - {\alpha}}$, where ${R_k}$ denote the distance between the ${k^{{\rm{th}}}}$ LEO satellite and the GU, $c$ is the light speed, ${{f_c}}$ denotes carrier frequency of the signal, and ${\alpha}$ is the path loss exponent. 

\subsubsection{Fading model}
In a satellite-to-ground (and ground-to-satellite) channel, the SR fading model, widely validated in satellite communication systems, is adopted, with the cumulative distribution function (CDF) of the channel gain ${\left| H \right|^2}$, described as \cite{CDFofSR2, pdfofsr} 
\begin{equation}\label{CDFofSR}
{F_{{{\left| {H} \right|}^2}}}(x)  = 1 - \mu \sum\limits_{n = 0}^{m - {\rm{1}}} {{\cal D}  \sum\limits_{l = 0}^n {\frac{{n!}}{{l!}}}} {\cal J}{ {x^l}{{\rm{e}}^{ - \left( {\beta  - \delta } \right)x}}},
\end{equation}
\noindent where ${\cal D}{\rm{ = }}\frac{{{{\left( {1 - m} \right)}_n}{{\left( { - \delta } \right)}^n}}}{{{{\left( {n!} \right)}^2}}}$, ${\cal J}{\rm{ = }}{\left( {\beta  - \delta } \right)^{ - \left( {n + 1 - l} \right)}}$. Moreover, $\mu  = \frac{1}{{2{{b_{\rm{s}}}}}}{\left( {\frac{{2{{b_{\rm{s}}}}m}}{{2{{b_{\rm{s}}}}m + \Omega }}} \right)^m}$, $\delta  = \frac{1}{{2{{b_{\rm{s}}}}}}\left( {\frac{\Omega }{{2{{b_{\rm{s}}}}m + \Omega }}} \right)$, $\beta  = \frac{1}{{2{{b_{\rm{s}}}}}}$ with ${\rm{\Omega }}$, $m$ and ${{b_{\rm{s}}}}$ being the average power of the line-of-sight (LoS) component, the Nakagami parameter and the half average power of the multi-path components except the LoS component. In addition, ${\left( x \right)_n}$ is the Pochhammer symbol. Accordingly, the probability density function (PDF) of the channel gain ${\left| H \right|^2}$ under SR fading is expressed as \cite{pdfofsr}
\begin{equation}\label{pdfofSR}
\begin{aligned}
{f_{{{\left| H \right|}^2}}}\left( x \right) = \mu \sum\limits_{n = 0}^{m - 1} {{\cal D}{x^n}{{\mathop{\rm e}\nolimits} ^{ - \left( {\beta  - \delta } \right)x}}}.
\end{aligned}
\end{equation}

\subsection{Signal Model}
For the uplink signal transmission from the GU to the ${k^{{\rm{th}}}}$ LEO satellite, the received SNR at the ${k^{{\rm{th}}}}$ LEO satellite is expressed as
\begin{equation}\label{snru}
{\rm{SNR}}_k^{\rm{u}} = \frac{{{p_{\rm{u}}}G_k^{\rm{u}}L_k\left( {{R_k}} \right){\left| {H_k^{\rm{u}}} \right|^2}}}{{\sigma _{k}^2}}, 
\end{equation}
\noindent where ${p_{\rm{u}}}$ is the transmit power of the GU, $G_k^{\rm{u}}$ presents the uplink antenna gain. In the mLEO-CN, each satellite is equipped with a directional antenna that emits a main lobe, directing one beam towards the center of the Earth. The GU is located within the beam of each LEO satellite, while the GU has an antenna to communicate with LEO satellites. Consequently, $G_k^{\rm{u}} = G_{\rm{u}}^{\rm{t}}G_{{\rm{s}},k}^{\rm{r}}$, where $G_{\rm{u}}^{\rm{t}}$ is the transmit antenna gain of the GU, and $G_{{\rm{s}},k}^{\rm{r}}$ presents the receive antenna gain of the ${k^{{\rm{th}}}}$ LEO satellite. Moreover, ${\sigma _k^2}$ is the received noise power in the ${k^{{\rm{th}}}}$ channel, ${{\left| {H_k^{\rm{u}}} \right|}^2}$ denotes the power gain of the uplink channel from the GU to the ${k^{{\rm{th}}}}$ LEO satellite, and the channel gains of the $K$ uplink channels between the GU and $K$ LEO satellites are mutually independent random variables. According to \cite{onlynoise}, numerical results show that the interference power in satellite communication is significantly lower than the noise power. Additionally, the ample bandwidth in satellite communications and the employed narrow beam transmission modes can enable orthogonal frequency allocation, thereby reducing co-channel interference in the mLEO-CN \cite{onlynoise2}. Therefore, as with most existing studies \cite{xujiajie1, onlynoise4, onlynoise}, we assume a noise-dominated propagation environment where co-channel interference is considered negligible. Accordingly, this assumption is more appropriate for scenarios in which spectrum resources are relatively abundant and communication reliability is prioritized over aggressive frequency reuse.

Accordingly, for the downlink signal transmission from the ${k^{{\rm{th}}}}$ LEO satellite to the GU, the received SNR at the GU is given by
\begin{equation}\label{snrd}
{\rm{SNR}}_k^{\rm{d}} = \frac{{{p_{\rm{s}}}G_k^{\rm{d}}L_k\left( {{R_k}} \right){{\left| {H_k^{\rm{d}}} \right|}^2}}}{{\sigma _{k}^2}},
\end{equation}
\noindent where ${p_{\rm{s}}}$ is the transmit power of each LEO satellite, $G_k^{\rm{d}}$ is the downlink antenna gain, which is expressed as $G_k^{\rm{d}} = G_{{\rm{s}},k}^{\rm{t}}G_{\rm{u}}^{\rm{r}}$, where $G_{{\rm{s}},k}^{\rm{t}}$ is the transmit antenna gain of the ${k^{{\rm{th}}}}$ LEO satellite, $G_{\rm{u}}^{\rm{r}}$ is the receive antenna gain of the GU. Similarly, ${{\left| {H_k^{\rm{d}}} \right|}^2}$ represents the power gain of the downlink channel from the ${k^{{\rm{th}}}}$ LEO satellite to the GU, and the channel gains of the $K$ downlink channels between $K$ LEO satellites and the GU are mutually independent random variables.

\subsection{Distance distribution}
To facilitate the derivation of subsequent analytical expressions, we adopt the following distance distribution \cite{PDFofA11}.
\begin{Lemma}
\label{Lemma-1}
The PDF of distance ${R_k}$ between the ${k^{{\rm{th}}}}$ LEO satellite and the GU is given by
\begin{equation}\label{pdfrd}
\begin{aligned}
{f_{{R_k}}}\left( {{r_k}} \right) = \frac{1}{\Xi }\sum\limits_{l = 0}^{k - 1} {\sum\limits_{n = 0}^{N - k} {\frac{\psi }{\Theta }} } r_k^{2n + 2l + 1},
\end{aligned}
\end{equation}
\noindent for ${r_{\min }} \le {r_k} \le {r_{\max }}$ and $k$ $\in \left\{ {1,2,3,...,K} \right\}$, where ${r_{\min }}$ is the constellation altitude, ${r_{\max }}$ denotes the maximum possible distance between the GU and an visible satellite, $\Theta  = {\left( {\frac{{ - r_{\min }^2}}{{4{r_{\rm{e}}}{r_a}}}} \right)^{l - k + 1}}{\left( {1 + \frac{{r_{\min }^2}}{{4{r_{\rm{e}}}{r_a}}}} \right)^{n + k - N}}{\left( {4{r_{\rm{e}}}{r_a}} \right)^{n + l + 1}}$, $\Xi  = \sum\limits_{j = k}^N {\left( {\begin{array}{*{20}{c}}
N\\
j
\end{array}} \right)} {\left( {\frac{{r_{\max }^2 - r_{\min }^2}}{{4{r_{\rm{e}}}{r_a}}}} \right)^j}{\left( {1 - \frac{{r_{\max }^2 - r_{\min }^2}}{{4{r_{\rm{e}}}{r_a}}}} \right)^{N - j}}$, and $\psi  = 2N\left( {\begin{array}{*{20}{c}}
{N - 1}\\
{k - 1}
\end{array}} \right)\left( {\begin{array}{*{20}{c}}
{N - k}\\
n
\end{array}} \right)\left( {\begin{array}{*{20}{c}}
{k - 1}\\
l
\end{array}} \right){\left( { - 1} \right)^n}$.
\end{Lemma}

\section{Performance Analysis}
\subsection{Exact Overall Transmission Success Probability}
The transmission success probability is defined as the probability that the received SNR is larger than a predefined threshold $\tau$, expressed as $P\left( \tau  \right)  \buildrel \Delta \over = \mathbb{P} \left( {{\rm{SNR}} > \tau } \right)$. For the end-to-end signal transmission, which includes the uplink from the GU to the ${k^{{\rm{th}}}}$ LEO satellite and the downlink from the ${k^{{\rm{th}}}}$ LEO satellite to the GU, the end-to-end transmission success probability represents the probability that the received SNR at both the LEO satellite and the GU satisfies the predefined threshold $\tau$, which is expressed as
\begin{equation}\label{cov-1}
   {P_k}\left( \tau  \right)  \buildrel \Delta \over = \mathbb{P} \left( {{\mathop{\rm SNR}\nolimits} _k^{\rm{u}} > \tau ,{\mathop{\rm SNR}\nolimits} _k^{\rm{d}} > \tau } \right),
\end{equation}
\noindent where ${\rm{SNR}}_k^{\rm{u}}$ and ${\rm{SNR}}_k^{\rm{d}}$ denote the received SNR at the ${k^{{\rm{th}}}}$ LEO satellite and the GU, respectively.

Under the adopted serial communication scheme, the GU probes serving satellites sequentially in the prescribed order, and a subsequent satellite is activated only if all preceding satellites fail to satisfy the target SNR threshold. Accordingly, the overall transmission success probability for $K$ serving satellites can be expressed as
\begin{equation}
\begin{aligned}
{P_{\rm{s}}}\left( \tau  \right)
=\, &P_1\left( \tau  \right)
+\left(1-P_1\left( \tau  \right)\right)P_2\left( \tau  \right) \\
&+\left(1-P_1\left( \tau  \right)\right)\left(1-P_2\left( \tau  \right)\right)P_3\left( \tau  \right) \\
&+\cdots+
\left(\prod_{j=1}^{K-1}\left(1-P_j\left( \tau  \right)\right)\right)P_K\left( \tau  \right)\\&\Leftrightarrow 1-\prod_{k=1}^K(1-P_k\left( \tau  \right)),
\end{aligned}
\end{equation}
\noindent where ${P_k}\left( \tau  \right)$ denotes the transmission success probability through the ${k^{{\rm{th}}}}$ LEO satellite. By De Morgan's law, the above serial-form expression is mathematically equivalent to a compact complement form, i.e., one minus the probability that all $K$ serving satellites fail.



To account for the effect of satellite failure, the end-to-end transmission success probability for each satellite should be multiplied by a factor $\left( {1{\rm{ - }}q_k} \right)$ representing the probability of the ${k^{{\rm{th}}}}$ LEO satellite being operational. Accordingly, by incorporating the satellite failure probability into the equivalent complement-form representation, the overall transmission success probability under the adopted serial communication scheme is given in Theorem \ref{thm1}.

\begin{thm}
\label{thm1}
For the serial communication supported by $K$ LEO satellites, the overall transmission success probability is expressed as
\begin{equation}\label{thm1-1}
 {P_{\rm{s}}}\left( \tau  \right) = 1 - \prod\limits_{k = 1}^K {\left( {1 -\overline q_k  {P_k}\left( \tau  \right)} \right)}, 
\end{equation}
\noindent where $\overline q_k  = 1 - q_k$ with $q_k$ being the failure probability of the ${k^{{\rm{th}}}}$ LEO satellite. Moreover, ${P_k}\left( \tau  \right)$ is expressed in (\ref{thm1-2}), at the top of the next page, where ${\cal G}{\rm{ = }}\frac{{\sigma _k^2{{\left( {4\pi {f_c}} \right)}^2}}}{{{c^2}}}$. In addition, ${f_{{R_k}}}\left( {{r_k}} \right)$ is the PDF of the distance between the ${k^{{\rm{th}}}}$ LEO satellite and the GU, as given in Lemma \ref{Lemma-1}. 
\begin{figure*}
\begin{equation}\label{thm1-2}
\begin{aligned}
 {P_k}\left( \tau  \right) = &\int_{{r_{\min }}}^{{r_{\max }}} {\left( {\mu \sum\limits_{n = 0}^{m - 1} {\cal D} \sum\limits_{l = 0}^n {\cal J} {{\left( \frac{{\tau{\cal G} r_k^\alpha }}{{{p_{\rm{u}}}G_k^{\rm{u}}}} \right)}^l}\exp \left( { - \frac{{ \left( {\beta  - \delta } \right)\tau{\cal G}r_k^\alpha }}{{{p_{{\rm{u}}}}G_k^{\rm{u}}}}} \right)} \right.} \left. {  \mu \sum\limits_{n = 0}^{m - 1} {\cal D} \sum\limits_{l = 0}^n {\cal J} {{\left( \frac{{\tau{\cal G}r_k^\alpha }}{{{p_{{\rm{s}}}}G_k^{\rm{d}}}} \right)}^l}\exp \left( { - \frac{{\left( {\beta  - \delta } \right)\tau{\cal G}r_k^\alpha }}{{{p_{{{\rm{s}}}}}G_k^{\rm{d}}}}} \right)} \right)\\
 &\times{f_{{R_k}}}\left( {{r_k}} \right)d{r_k}.
\end{aligned}
\end{equation}\hrulefill
\end{figure*}
\end{thm}

The \textit{Proof} of Theorem \ref{thm1} is given in Appendix A.

\subsection{Asymptotic Overall Transmission Success Probability}
Due to the complexity of the CDF of the SR fading model and the PDF of the distance between the ${k^{{\rm{th}}}}$ LEO satellite and the GU, deriving a closed-form expression for the exact overall transmission success probability poses significant challenges. Therefore, to gain deeper insights,  we derive asymptotic expressions for the overall transmission success probability in specific cases.

\emph{Case 1:}  For the high SNR regime $\left( {{\rm{SNR}} \to \infty } \right)$, we have $x \to 0$. The asymptotic CDF of the SR fading model can be expressed as $F_{{{\left| H \right|}^2}}^\infty \left( x \right) \approx \mu  x$ \cite{can2023outage}.
\begin{equation}\label{appCDF1}
\begin{aligned}
     F_{{{\left| H \right|}^2}}^\infty \left( x \right) \approx \mu  x.
    \end{aligned}
\end{equation} 
\begin{Col}\label{col2}
The end-to-end transmission success probability through the ${k^{{\rm{th}}}}$ LEO satellite at high SNR regime can be approximated as 
\begin{equation}\label{highSNR}
\begin{aligned}
    P_k^\infty \left( \tau  \right) \approx& 1+\frac{{{\mu ^2}{\tau^2 {\cal G}^2}}}{{{p_{{\rm{u}}}}G_k^{\rm{u}}{p_{{\rm{s}}}} G_k^{\rm{d}}\Xi }}\sum\limits_{l = 0}^{k - 1} {\sum\limits_{n = 0}^{N - k} {\frac{\psi }{\Theta }\left( {\frac{{r_{\max }^{2\alpha  + \Delta} - r_{\min }^{2\alpha  + \Delta}}}{{2\alpha  + \Delta}}} \right)} }\\
    &-\frac{{{\mu }{\tau{\cal G}}}}{{{p_{{\rm{u}}}}G_k^{\rm{u}}\Xi }}\sum\limits_{l = 0}^{k - 1} {\sum\limits_{n = 0}^{N - k} {\frac{\psi }{\Theta }\left( {\frac{{r_{\max }^{\alpha  + \Delta} - r_{\min }^{\alpha  + \Delta}}}{{\alpha  + \Delta}}} \right)} }\\
    &-\frac{{{\mu }{\tau{\cal G}}}}{{{p_{{\rm{s}}}} G_k^{\rm{d}}\Xi }}\sum\limits_{l = 0}^{k - 1} {\sum\limits_{n = 0}^{N - k} {\frac{\psi }{\Theta }\left( {\frac{{r_{\max }^{\alpha  + \Delta} - r_{\min }^{\alpha  + \Delta}}}{{\alpha  + \Delta}}} \right)} },
\end{aligned}
\end{equation}
\noindent where $\Delta  = 2n + 2l + 2$.
\end{Col} 

The \textit{Proof} of Corollary 1 is given in Appendix B.

\emph{Case 2:} The SR fading model can be approximated by a Gamma distribution, $G \sim \Gamma \left( {{\alpha _s},{\beta _s}} \right)$, with shape and scale parameters ${{\alpha _s}}$ and ${{\beta _s}}$ \cite{appsr1, 24-47}. The corresponding approximate CDF is given by
\begin{equation} \label{cdfappsr3}
{F_{{{\left| {H} \right|}^2}}}\left( x \right) \approx {{\gamma \left( {{\alpha _s},\frac{x}{{{\beta _s}}}} \right)} \mathord{\left/
 {\vphantom {{\gamma \left( {{\alpha _s},\frac{x}{{{\beta _s}}}} \right)} {\Gamma \left( {{\alpha _s}} \right)}}} \right.
 \kern-\nulldelimiterspace} {\Gamma \left( {{\alpha _s}} \right)}},
\end{equation}
\noindent where ${\alpha _s} = \frac{{m{{\left( {2{{b_{\rm{s}}}} + {\rm{\Omega }}} \right)}^2}}}{{4m{b_{\rm{s}}}^2 + 4m{{b_{\rm{s}}}}{\rm{\Omega }} + {{\rm{\Omega }}^2}}}$ and ${\beta _s} = \frac{{4m{b_{\rm{s}}}^2 + 4m{{b_{\rm{s}}}}{\rm{\Omega }} + {{\rm{\Omega }}^2}}}{{m\left( {2{{b_{\rm{s}}}} + {\rm{\Omega }}} \right)}}$.
\begin{Col}
Based on this approximation, the end-to-end transmission success probability can be simplified as 
\begin{equation} \label{covapp3}
\begin{aligned}
P_k^{{\rm{app}}}\left( \tau  \right)  \approx&   \frac{1}{\Xi }\sum\limits_{l = 0}^{k - 1} {\sum\limits_{n = 0}^{N - k} {\sum\limits_{{t_1} = 0}^{\widetilde {{\alpha _s}} - 1} {\sum\limits_{{t_2} = 0}^{\widetilde {{\alpha _s}} - 1} {\frac{\psi }{\Theta }\frac{{{{\left( {\frac{\tau{\cal G}}{{{p_{{\rm{u}}}}G_k^{\rm{u}}}}} \right)}^{{t_1}}}{{\left( {\frac{\tau{\cal G}}{{{p_{{\rm{s}}}}G_k^{\rm{d}}}}} \right)}^{{t_2}}}}}{{\alpha \Gamma \left( { {t_1} + 1} \right)\Gamma \left( { {t_2} + 1} \right)}}} } } }\\
&\times\left( \frac{{r_{{\rm{min}}}^{\rm{B}}\Gamma \left( {\frac{{\rm{B}}}{\alpha }{\rm{A}}r_{{\rm{min}}}^\alpha } \right)}}{{{{\left( {{\rm{A}}r_{{\rm{min}}}^\alpha } \right)}^{\frac{{\rm{B}}}{\alpha }}}}}-\frac{{r_{{\rm{max}}}^{\rm{B}}\Gamma \left( {\frac{{\rm{B}}}{\alpha }{\rm{A}}r_{{\rm{max}}}^\alpha } \right)}}{{{{\left( {{\rm{A}}r_{{\rm{max}}}^\alpha } \right)}^{\frac{{\rm{B}}}{\alpha }}}}} \right).
\end{aligned}
\end{equation}
\noindent where ${\rm A}{\rm{ = }}\left( {\frac{{\cal G}}{{{p_{{\rm{u}}}}G_k^{\rm{u}}}} + \frac{{\cal G}}{{{p_{{\rm{s}}}}G_k^{\rm{d}}}}} \right)$, ${\rm B}=\left( {{t_1} + {t_2}} \right)$, and $\widetilde {{\alpha _s}}$ represents an approximate integer of ${{\alpha _s}}$, that is $\widetilde {{\alpha _s}} = \left\lfloor {{\alpha _s}} \right\rfloor $ or $\left\lceil {{\alpha _s}} \right\rceil$.
\end{Col}

The \textit{Proof} of Corollary 2 is given in Appendix C.

\subsection{Cluster-based Approximation Approach for Complexity Reduction}
The aforementioned asymptotic expression only simplifies the end-to-end transmission success probability for each LEO satellite individually. However, the computational complexity of the overall success probability still increases dramatically as the number of serving satellites grows. To address this, we propose a cluster-based approximation approach to estimate the overall transmission success probability for $K > 2$. Instead of evaluating each satellite individually, the core idea of the cluster-based approximation is to divide the serving LEO satellites into multiple clusters and represent the performance of each cluster using its central satellite. This approach simplifies the cumulative product of each satellite’s performance by approximating it as the power of the central satellite’s performance, effectively reducing computational complexity without significantly sacrificing accuracy.
\begin{pro}
The overall transmission success probability based on the cluster-based approximation approach is expressed as
\begin{equation}
\begin{aligned}\label{cluster}
P_s^{{\rm{app}}}(\tau ) & \approx  1 - \left( {\prod\limits_{j = 1}^{\left\lfloor {{K \mathord{\left/
 {\vphantom {K {{\cal K} }}} \right.
 \kern-\nulldelimiterspace} {{\cal K}}}} \right\rfloor } {{{\left( {1 - \textcolor{black}{\overline q_k} {P_{ \left( {j - 1} \right){\cal K}+\left\lfloor {\frac{{{\cal K} + 1}}{2}} \right\rfloor }}(\tau )} \right)}^{{\cal K}}}} } \right)\\
 &\quad \times {\left( {1 -\textcolor{black}{\overline q_k}  {P_{\left\lfloor {{K \mathord{\left/
 {\vphantom {K {{\cal K}}}} \right.
 \kern-\nulldelimiterspace} {\cal K}}} \right\rfloor  + \left\lfloor {\frac{{\bmod \left( {K,{\cal K}} \right) + 1}}{2}} \right\rfloor }}(\tau )} \right)^{\bmod \left( {K,{\cal K}} \right)}},
    \end{aligned}
\end{equation}
\noindent $for K>2$, where ${\cal K} ({\cal K} \le K,{\cal K} \in {{\mathbb{N}}^ + })$ is the cluster size, and $\left\lfloor  \cdot  \right\rfloor$ denotes the floor operation, and $\bmod \left( {u,v} \right)\left( {u, v \in {\mathbb{N}^ + }} \right)$ denotes the remainder of the Euclidean division of $u$ by $v$.
\end{pro}
{\it{Proof:} } Since the performance of all LEO satellites in each cluster is represented by the power of the central satellite's performance, thus ${P_k}\left( \tau  \right)$ in (\ref{thm1-1}) is replaced by ${P_{\left( {j - 1} \right){\cal K} + \left\lfloor {\frac{{{\cal K} + 1}}{2}} \right\rfloor }}$ in the first line of (\ref{cluster}) to denote the central satellite of each cluster. Then, the ${\cal K}^{{\rm{th}}}$ power of this item is used to represent the performance of all satellites in each cluster. Subsequently, the $K$ accumulative multiplications are reduced to $\left\lfloor {{K \mathord{\left/
 {\vphantom {K {\cal K}}} \right.
 \kern-\nulldelimiterspace} {\cal K}}} \right\rfloor $ to reflect the cumulative effect of $K$ central satellites. In addition, since $K$ is typically not an exact multiple of  ${\cal K}$, a similar way is employed in the second line of (\ref{cluster}) to account for the remainder of LEO satellites. \hfill $\square$

To compare the mathematical complexity of the original expression in (\ref{thm1-1}) with the proposed approximation approach in (\ref{cluster}), we quantify them through the number of integrals and multiplications \cite{Dong2022when}. Since the end-to-end transmission success probability for each LEO satellite involves one integral, in (\ref{thm1-1}), the overall transmission success probability for the serial communication with $K$ LEO satellites consists of $K$ integrals and $K$ multiplications. In contrast, in the proposed cluster-based approximation approach, the number of integrals is reduced to $\cal K$ and the number of multiplications decreases to $\frac{K}{{\cal K}}$. Therefore, compared to the original expression, the proposed approach significantly reduces mathematical complexity.

\subsection{Average Transmission Rate Capacity}
The performance measure of the transmission rate capacity is defined as the ergodic capacity from the Shannon-Hartley theorem, given by ${C} \buildrel \Delta \over = B{\log _2}\left( {1 + {\rm{SN}}{{\rm{R}}}} \right)$, where $B$ is the bandwidth. For the end-to-end transmission involving the uplink from the GU to the ${k^{{\rm{th}}}}$ LEO satellite and the downlink from the ${k^{{\rm{th}}}}$ LEO satellite to the GU, the end-to-end transmission rate capacity through the ${k^{{\rm{th}}}}$ LEO satellite depends on the minimum of the received SNR at the ${k^{{\rm{th}}}}$ LEO satellite and the received SNR at the GU, which is given by
\begin{equation}
\begin{aligned}
{C_k}  \buildrel \Delta \over =  B{\log _2}\left( {1 + \min \left( {{\rm{SNR}}_k^{\rm{u}},{\rm{SNR}}_k^{\rm{d}}} \right)} \right).
\end{aligned}
\end{equation}

In the mLEO-CN, the average transmission rate capacity is defined as the effective end-to-end transmission rate capacity across $K$ LEO satellites, starting from the first, where each subsequent satellite is included only if all previous LEO satellites cannot meet the predefined SNR threshold.
\begin{thm}
\label{thm2}
The average transmission rate capacity for the serial communication supported by $K$ LEO satellites mLEO-CN is expressed as
\begin{equation}
 {C_{\rm{s}}} = \left\{ {\begin{array}{*{20}{l}}
{{ \overline q_k  {C_1}},\qquad \qquad \qquad \quad \qquad \qquad K = 1,}\\
{ { \overline q_k {C_1} + \sum\limits_{k = 2}^K { \overline q_k {C_k}\prod\limits_{j = 1}^{k - 1} {\left( {1 - {P_j}\left( \tau  \right)} \right)} } } ,  K \ge 2,}
\end{array}} \right.
  \end{equation}
\noindent where ${P_j}\left( \tau  \right)$ denotes the end-to-end transmission success probability through the ${j^{{\rm{th}}}}$ LEO satellite, and ${C_1}$ and ${C_k}$ are the transmission rate capacities through the first and ${k^{{\rm{th}}}}$ LEO satellites, respectively, ${C_k}$ is expressed in (\ref{capi}), on top of this page. In (\ref{capi}), $\mathbbm{1}\left(  \cdot  \right)$ denotes an indicator function, ${f_{\left| {H_k^{\rm{u}}} \right|^2}}\left( {{{h_k^{\rm{u}}}}} \right)$ and ${f_{\left| {H_k^{\rm{d}}} \right|^2}}\left( {{h_k^{\rm{d}}}} \right)$ are the PDFs of the SR fading in the uplink and downlink of the ${k^{{\rm{th}}}}$ channel, respectively, as given in (\ref{pdfofSR}).
\begin{figure*}[t]
\begin{equation} \label{capi}
\begin{aligned}
{C_k}= &B\int_{{r_{\min }}}^{{r_{\max }}} {\int_0^\infty  {\int_0^\infty  {{{\log }_2}\left( {1 + \min \left( {\frac{{{p_{{\rm{u}}}}G_k^{\rm{u}}{{\left| {h_k^{\rm{u}}} \right|}^2}r_k^{ - \alpha }}}{{\cal G}},\frac{{{p_{{\rm{s}}}}G_k^{\rm{d}}{{\left| {h_k^{\rm{d}}} \right|}^2}r_k^{ - \alpha }}}{{\cal G}}} \right)} \right)} } }\mathbbm{1} \left( {\frac{{{p_{{\rm{u}}}}G_k^{\rm{u}}{{\left| {h_k^{\rm{u}}} \right|}^2}r_k^{ - \alpha }}}{{\cal G}} > \tau} \right)\\
 &\times \mathbbm{1} \left( {\frac{{{p_{{\rm{s}}}}G_k^{\rm{d}}{{\left| {h_k^{\rm{d}}} \right|}^2}r_k^{ - \alpha }}}{{\cal G}} > \tau} \right){f_{\left| {H_k^{\rm{u}}} \right|^2}}\left( {{{h_k^{\rm{u}}}}} \right){f_{\left| {H_k^{\rm{d}}} \right|^2}}\left( {{h_k^{\rm{d}}}} \right)d{{h_k^{\rm{u}}}}d{{h_k^{\rm{d}}}}f_{R_k}\left( {{r_k}} \right)d{r_k}.
 \end{aligned}
\end{equation}\hrulefill
\end{figure*}
\end{thm}

The \textit{Proof} of Theorem \ref{thm2} is given in Appendix D.

\subsection{End-to-End Delay}
In the mLEO-CN, the end-to-end delay is considered to be composed of transmission delay and the handover delay. The transmission delay characterizes the time required for message delivery between the GU and the ${k^{{\rm{th}}}}$ LEO satellite, expressed as the sum of the uplink delay $T_k^{\rm{u}}$ from the GU to the ${k^{{\rm{th}}}}$ LEO satellite and the downlink delay $T_k^{\rm{d}}$ from the ${k^{{\rm{th}}}}$ LEO satellite to the GU. Moreover, when employing the serial communication scheme, a handover delay arises when the current LEO satellite cannot satisfy the predefined SNR threshold, necessitating a switch to an alternative LEO satellite. Notably, as the number of switches increases, the cumulative handover delay becomes non-negligible. For analytical tractability, the handover delay  ${T_{{\rm{ho}}}}$ per switch is assumed to be a constant and the cumulative handover delay incurred when switching from the first LEO satellite to the ${k^{{\rm{th}}}}$ LEO satellite is calculated as $\left( {k - 1} \right) {T_{{\rm{ho}}}}$. Thus, the total delay through the ${k^{{\rm{th}}}}$ LEO satellite is 
\begin{equation}\label{delayk}
    T_k  \buildrel \Delta \over =  T_k^{\rm{u}} + T_k^{\rm{d}}+\left( {k - 1} \right) {T_{{\rm{ho}}}}.
\end{equation}
Therefore, by employing the serial communication scheme involving $K$ serving LEO satellites, the end-to-end delay is defined as the average of the total delay through any of the $K$ LEO satellites that successfully transmit the message data.
\begin{thm}
\label{thm3}
The end-to-end delay for the serial communication supported by $K$ LEO satellites in the mLEO-CN is expressed as
\begin{equation}\label{delay}
 {T_{\rm{s}}} = \left\{ {\begin{array}{*{20}{l}}
{\frac{T_1}{{{P_1}\left( \tau  \right)}},\qquad \qquad  \qquad  \qquad  K = 1,}\\
{ \frac{{T_1} + \sum\limits_{k = 2}^K {{T_k}  \prod\limits_{j = 1}^{k - 1} {\left( {1 - {P_j}\left( \tau  \right)} \right)} } }{{1 - \prod\limits_{k = 1}^K {\left( {1 - {P_k}\left( \tau  \right)} \right)} }},\quad  K \ge 2,}
\end{array}} \right.
\end{equation}
\noindent where ${T_1}$ and ${T_k}$ denote the total delay through the first and the ${k^{{\rm{th}}}}$ LEO satellites, respectively. Note that, ${T_k}$ is expressed in (\ref{delayTk}), on top of the next page, where $\omega$ is the length of the message packet.
\begin{figure*}[t]
\begin{equation} \label{delayTk}
\begin{aligned}
{T_k}= &\int_{{r_{\min }}}^{{r_{\max }}} {\int_0^\infty \int_0^\infty {\frac{\omega \mathbbm{1}\left( {\frac{{{p_{{\rm{u}}}}G_k^{\rm{u}}{{\left| {h_k^{\rm{u}}} \right|}^2}r_k^{ - \alpha }}}{{\cal G}}> \tau} \right)}{{B{{\log }_2}\left( {1 + \frac{{{p_{{\rm{u}}}}G_k^{\rm{u}}{{\left| {h_k^{\rm{u}}} \right|}^2}r_k^{ - \alpha }}}{{\cal G}}} \right)}} }}{\mathbbm{1}\left( {\frac{{{p_{{\rm{s}}}}G_k^{\rm{d}}{{\left| {h_k^{\rm{d}}} \right|}^2}r_k^{ - \alpha }}}{{\cal G}}> \tau} \right)}{f_{\left| {H_k^{\rm{u}}} \right|^2}}\left( h_k^{\rm{u}} \right){f_{\left| {H_k^{\rm{d}}} \right|^2}}\left( h_k^{\rm{d}} \right){ {f_{{R_k}}}\left( {{r_k}} \right)dh_k^{\rm{u}}dh_k^{\rm{d}}d{r_k}}\\
&+ \int_{{r_{\min }}}^{{r_{\max }}} {\int_0^\infty \int_0^\infty {\frac{\omega \mathbbm{1}\left( {\frac{{{p_{{\rm{u}}}}G_k^{\rm{u}}{{\left| {h_k^{\rm{u}}} \right|}^2}r_k^{ - \alpha }}}{{\cal G}}> \tau} \right)}{{B{{\log }_2}\left( {1 + \frac{{{p_{{\rm{s}}}}G_k^{\rm{d}}{{\left| {h_k^{\rm{d}}} \right|}^2}r_k^{ - \alpha }}}{{\cal G}}} \right)}}}}{\mathbbm{1}\left( {\frac{{{p_{{\rm{s}}}}G_k^{\rm{d}}{{\left| {h_k^{\rm{d}}} \right|}^2}r_k^{ - \alpha }}}{{\cal G}}> \tau} \right)}{{f_{\left| {H_k^{\rm{u}}} \right|^2}}\left( h_k^{\rm{u}} \right){f_{\left| {H_k^{\rm{d}}} \right|^2}}\left( h_k^{\rm{d}} \right) {f_{{R_k}}}\left( {{r_k}} \right)dh_k^{\rm{u}}dh_k^{\rm{d}}d{r_k}}\\
&+ \int_{{r_{\min }}}^{{r_{\max }}} {\int_0^\infty \int_0^\infty  {\left( {k - 1} \right){T_{{\rm{ho}}}}}}{{\mathbbm{1}\left( {\frac{{{p_{{\rm{u}}}}G_k^{\rm{u}}{{\left| {h_k^{\rm{u}}} \right|}^2}r_k^{ - \alpha }}}{{\cal G}}> \tau} \right) \mathbbm{1}\left( {\frac{{{p_{{\rm{s}}}}G_k^{\rm{d}}{{\left| {h_k^{\rm{d}}} \right|}^2}r_k^{ - \alpha }}}{{\cal G}}> \tau} \right)} {f_{\left| {H_k^{\rm{u}}} \right|^2}}\left( h_k^{\rm{u}} \right){f_{\left| {H_k^{\rm{d}}} \right|^2}}\left( h_k^{\rm{d}} \right){f_{{R_k}}}\left( {{r_k}} \right)}\\
&\times dh_k^{\rm{u}}dh_k^{\rm{d}}d{r_k}.
 \end{aligned}
\end{equation}\hrulefill
\end{figure*}
\end{thm}

The \textit{Proof} of Theorem \ref{thm3} is given in Appendix E.

\subsection{Optimal Number of Serving LEO Satellites}
Although the serial communication supported by multiple LEO satellites can improve the overall transmission success probability, it also introduces extra handover delays at each switch, resulting in a significant increase in the end-to-end delay. To balance the overall performance of mLEO-CN, based on the derived analytical results presented in Section \uppercase\expandafter{\romannumeral3}, we formulate an optimization problem to maximize the ratio of the overall transmission success probability to the end-to-end delay, as given in (\ref{opt0}) at the top of the next page. Our objective is to find the optimal number of serving LEO satellites $K^\star$ under the maximum $\eta$, formulated as 
\begin{figure*}
{\color{black}\begin{equation}\label{opt0}
\eta = \frac{P_s(\tau)}{T_s}
=
\begin{cases}
\dfrac{\bar q_k{{{\left( {{P_1}\left( \tau  \right)} \right)}^2}}}{T_1}, & K=1, \\[1.2ex]
\dfrac{
\left(1-\prod\limits_{k=1}^{K}\left(1-\bar q_k{P_k}(\tau)\right)\right)
\left(1-\prod\limits_{k=1}^{K}\left(1-P_k(\tau)\right)\right)
}{
T_1+\sum\limits_{k=2}^{K} T_k \prod\limits_{j=1}^{k-1}\left(1-P_j(\tau)\right)
}, & K\ge 2.
\end{cases}
\end{equation}}\hrulefill
\end{figure*}
{\color{black}
\begin{subequations}
    \begin{align*}\label{op1}
    {\rm (\mathbf{P1}):}\quad
K^\star=\arg\max_{K} &\quad \eta\tag{22}\\
{{\rm{s}}{\rm{.t}}{\rm{.}}}\qquad &{{\Phi _{\rm{vis}}},{\Phi _s} \ne \phi }\tag{22a}\\
&{K \in {{{\mathbb{N}}^ + } }}, {K \le N}, \tag{22b}\\
&{{r_{\rm{l}}} \le {r_{\min }} \le {r_{\rm{u}}}},\tag{22c}\\
&{{P_s}\left( \tau  \right),{P_k}\left( \tau  \right),q \in \left[ {0,1} \right]},\tag{22d}
    \end{align*}   
\end{subequations}}

In problem ${\rm (\mathbf{P1})}$, constraint (22a) ensures that both the sets of visible satellites and serving satellites are non-empty. Constraint (22b) specifies that the constellation size and the number of serving LEO satellites must be positive integers and enforces that the number of serving LEO satellites does not exceed the constellation size. Additionally, constraint (22c) defines the permissible range for the constellation altitude, which varies from ${r_{\rm{l}}}$ to ${r_{\rm{u}}}$. Lastly, constraint (22d) ensures that both the overall transmission success probability and the end-to-end transmission success probability via the ${k^{{\rm{th}}}}$ LEO satellite fall within the range of zero to one. 

\begin{pro}\label{Prop1}
    The objective function $\eta$ is unimodal with respect to the number of serving satellites $K$.
\end{pro}
 
The proof of the Proposition \ref{Prop1} is given in Appendix F.


Given the above analysis and proof of unimodality, ${\rm (\mathbf{P1})}$ is still challenging to solve, since its objective function is derived from the stochastic-geometry analysis in Sections III-A and III-C and involves summations, products, and integrals, making it analytically intractable for conventional convex optimization methods. Furthermore, the decision variable, namely the number of serving satellites $K$, is integer-valued, and our objective is to find $K^\star$ under different network parameter settings. Therefore, directly applying continuous relaxation may yield non-integer solutions, which are not practically meaningful in the mLEO-CN. Since ${\rm (\mathbf{P1})}$ is a one-dimensional integer optimization problem, tractable search-based have been widely adopted in the existing studies on similar problems to efficiently obtain optimal solutions \cite{onlynoise4, Nearest-D-BPP1, Dwivedi2024Performance, Qu2026Coverage, Deng2025Satellite}. Motivated by these studies, we employ a one-dimensional search algorithm to solve ${\rm (\mathbf{P1})}$. Specifically, the core idea of the algorithm is to sequentially search over the number of serving satellites $K$ for a given constellation altitude and terminate the search once $\eta$ starts to decrease due to the unimodality of the objective function. The detailed steps are provided in Algorithm \ref{alg1}.
\begin{algorithm}[h]
\caption{\textcolor{black}{One-Dimensional Search Algorithm for Solving ${\rm (\mathbf{P1})}$}}
\label{alg1}
{\color{black}
\begin{algorithmic}[1]
\Require $\tau$, $r_{\min}$, and upper bound on the number of serving satellites $K_{\rm u}$
\Ensure maximum ratio of the overall transmission success probability to the end-to-end delay $\eta^\star$ and $K^\star$
\State Initialize $K \gets 1$ and compute $\eta(1)$
\For{$K$ from 2 to $K_{\rm u}$}
    \State Compute $\eta(K)$ using (21)
    \If{$\eta(K)<\eta(K-1)$}
        \State Set $K^\star \gets K-1$ and $\eta^\star \gets \eta(K-1)$
        \State \textbf{break}
    \EndIf
\EndFor
\If{the search reaches $K_{\rm u}$ without decrease}
    \State Set $K^\star \gets K_{\rm u}$ and $\eta^\star \gets \eta(K_{\rm u})$
\EndIf
\end{algorithmic}
}
\end{algorithm}

\section{Numerical Results}
\subsection{Parameter Settings}
In the simulation, the GU is positioned on the Earth at a specific point $\left({0,0,{r_{\rm{e}}}} \right)$ with Earth radius ${r_{\rm{e}}} = 6371 ~ {\rm{km}}$. The $N$ LEO satellite positions are created on a spherical layer at the altitude of ${r_{\min }}$ above the Earth's surface, constituting a BPP. All the LEO satellites are ranked from 1 to $N$ according to their distances to the GU. For simplicity, we set ${q_1} = {q_2} = ...{q_k} = q$. Unless otherwise specified, $N$ is set to 400, ${r_{\min }}$ is set to 1000 km, the predefined SNR threshold $\tau$ is set to -10 dB, ${r_{\rm{l}}}$ as well as ${r_{\rm{u}}}$ are set to 500 km and 1500 km, and ${K_{\rm{u}}}$ is set to 10. 

The transmit power of the GU is 3 W, and the transmit antenna gain and the receive antenna gain of the GU are set to 10 dBi and 0 dBi \cite{Nearest-D-BPP1}, respectively. For the LEO satellite, the transmit power is 10 W \cite{okati2020downlink}, the transmit antenna gain is set to 30 dBi, \cite{Nearest-D-BPP1}, and the receive antenna gain is set to  25 dBi \cite{Nearest-D-BPP1}. The parameters in the SR fading model $\left( {{{b_{\rm{s}}}},m,\Omega } \right)$ are set as $\left(0.4, 3, 0.2\right)$ \cite{CDFofSR2}. Moreover, the failure probability ${q}$ is set to 0 or 0.1 to represent the two states of each LEO satellite. The noise power spectral densities (PSDs) ${N_0}$ is set to $-$203 dBm/Hz \cite{Deng21Ultra}. The bandwidth of the channel in the space link is set as 250 MHz \cite{Liu2024Performance}. Therefore, by applying the relationship among noise PSD, bandwidth $B$, and noise power ${\sigma ^2}$, that is ${N_0}B = {\sigma ^2}$ \cite{onlynoise}, the noise power is $-$119 dBm. In addition, the path loss exponent $\alpha$ is set to 2.2, the carrier frequency ${f_c}$ of the signal is set to $30~{\rm{GHz}}$ and the handover delay ${T_{{\rm{ho}}}}$ at each switch is set to 20 ms \cite{jiang2023leo}.

\subsection{Results and Discussion} 
In all the figures, the lines correspond to the analytical results, while the markers indicate the outcomes from Monte Carlo simulations. We verify the accuracy of our derivations for the overall transmission success probability and the corresponding asymptotic expressions derived in Sections III-A and B, as shown in Fig. \ref{cov-T} to Fig. \ref{cov-app2}. Fig. \ref{cov-T} illustrates how the overall transmission success probability varies with the predefined SNR threshold, comparing theoretical and simulation results under failure probabilities of $q$=0 and $q$=0.1 for $K$=1, 2, 3. Theoretical results match the simulations, validating the accuracy of the derivations. It is observable that increasing the number of cooperative LEO satellites leads to a gradual enhancement in the overall transmission success probability. This is due to the fact that when the first LEO satellite suffers from a lower coverage performance due to factors such as channel fading, the employed serial communication scheme promotes additional LEO satellites to provide overlapping coverage for the GU, which in turn reduces the likelihood of outage. Therefore, the serial communication with multiple LEO satellites can effectively improve the end-to-end signal transmission performance compared to a single satellite-supported system. Moreover, the potential satellite failures significantly degrade the overall transmission success probability, particularly in the single-satellite scenario. For example, when $K$=1 and $q$=0, the transmission success probability ranges from 1 to 0.908 as the SNR threshold varies from -30 dB to -20 dB, whereas with $q$=0.1, it drops from 0.8 to 0.726. This reduction arises from the lack of alternative satellites to mitigate the impact of failure on the serving satellite. In contrast, when the failure probability $q$ =0.1, the transmission success probability is substantially higher in the cases with $K$=2 and $K$=3, compared to the single-satellite case with $K$=1. Thus, the serial communication with multiple LEO satellites effectively alleviates induced degradation, enhancing the reliability and robustness of LEO satellite communication networks. 
 \begin{figure}[t]	
 \centerline{\includegraphics[scale=0.55]{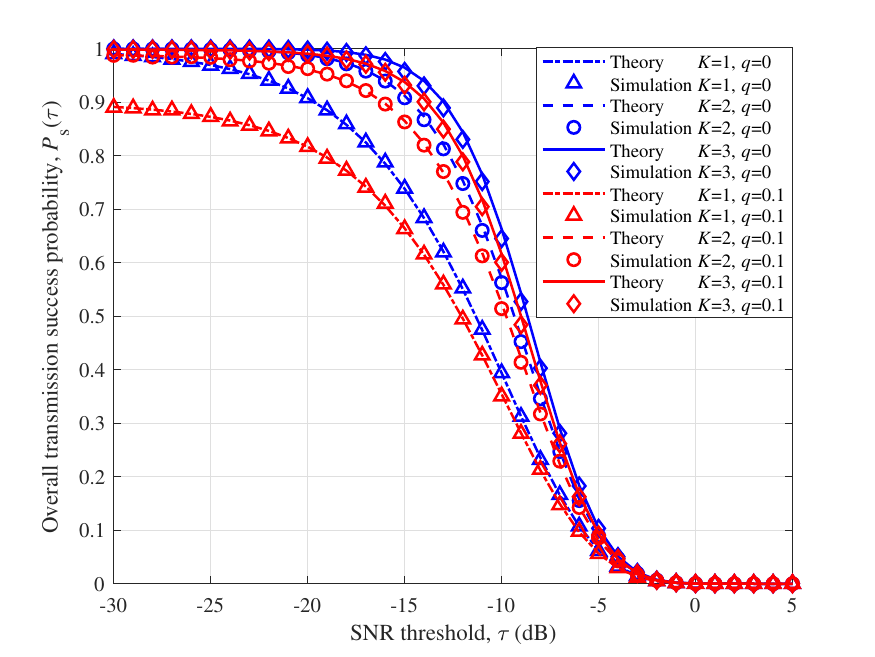}}
		\caption{Effect of predefined SNR threshold $\tau$ on the overall transmission success probability.}
		\label{cov-T}
 \end{figure}
 \begin{figure}[t]
		\centerline{\includegraphics[scale=0.55]{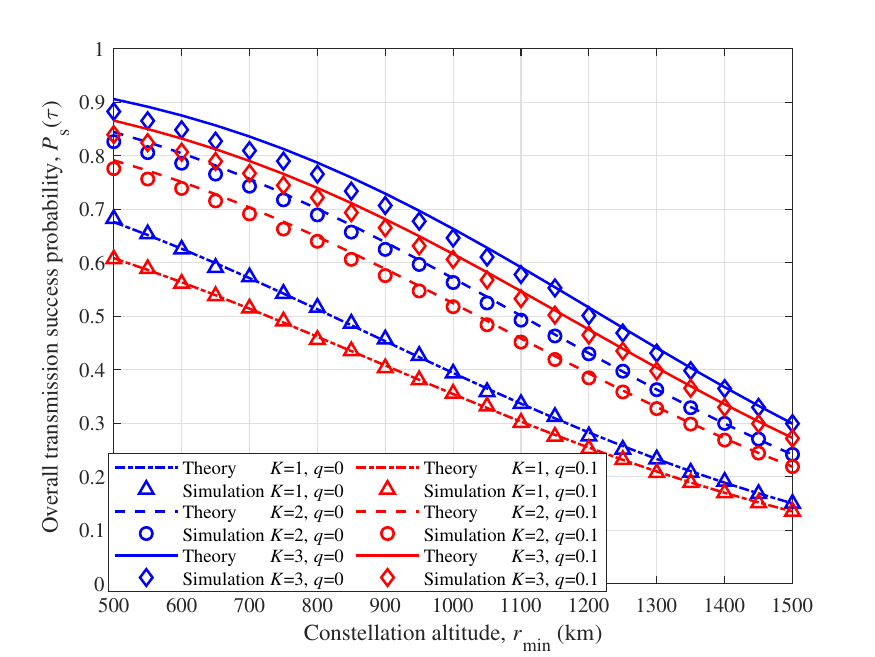}}
		\caption{Effect of the constellation altitude on the overall transmission success probability.}
		\label{cov-rmin}
 \end{figure}
  
Fig. \ref{cov-rmin} illustrates the effect of constellation altitude on the overall transmission success probability. As the constellation altitude increases, the overall transmission success probability declines progressively. This phenomenon arises due to the extended distance between the GU and each LEO satellite at higher altitudes, amplifying path loss in signal transmissions. Consequently, both the received SNR in the uplink and the downlink deteriorate, lowering the overall transmission success probability. It is important to highlight that a discrepancy between theoretical and simulation results is observed at lower constellation altitudes, particularly for $K$=2 and $K$=3. This discrepancy stems from the low degree of alignment between the employed PDF of the distance between the ${k^{{\rm{th}}}}$ LEO satellite and the GU and the actual distance distribution in practice. 
 \begin{figure}[t]
		\centerline{\includegraphics[scale=0.55]{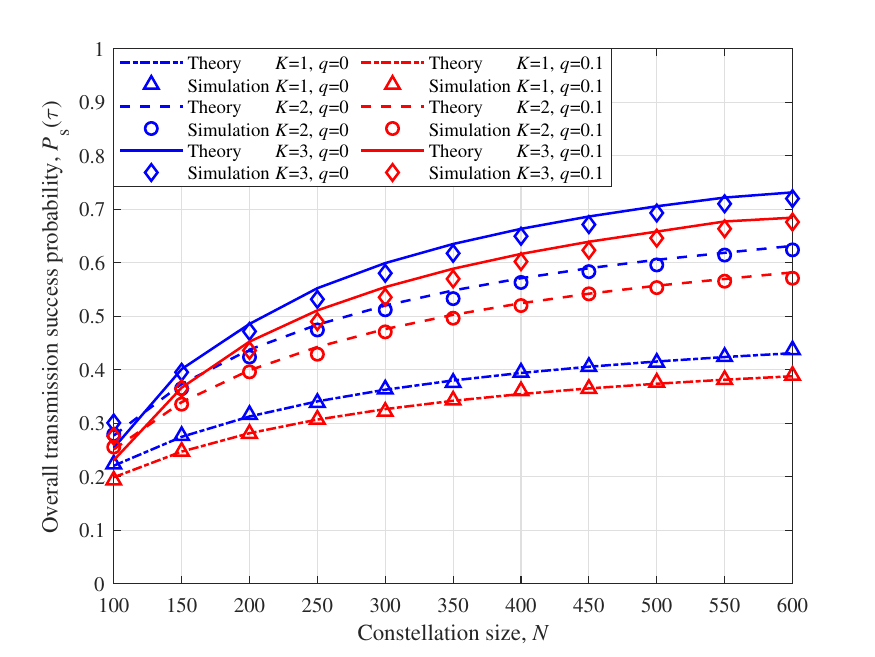}}
		\caption{Effect of the constellation size on the overall transmission success probability.}
		\label{cov-N}
 \end{figure}
   \begin{figure}[t] 
		\centerline{\includegraphics[scale=0.55]{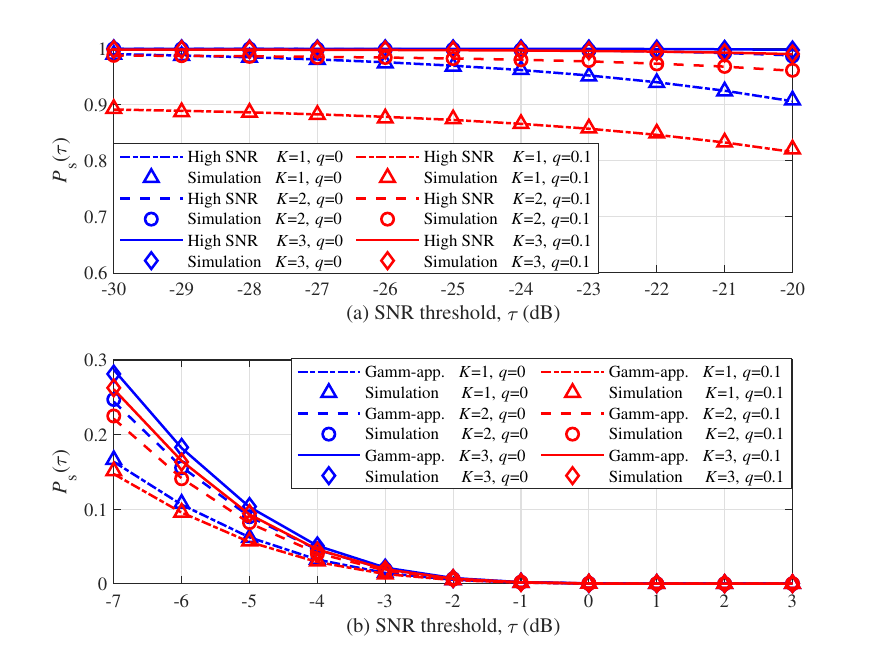}}
		\caption{Validation of asymptotic expressions on the end-to-end transmission success probability through the ${k^{{\rm{th}}}}$ LEO satellite under two special cases.}
		\label{cov-app1}
 \end{figure}
\begin{figure}[t]
		\centerline{\includegraphics[scale=0.55]{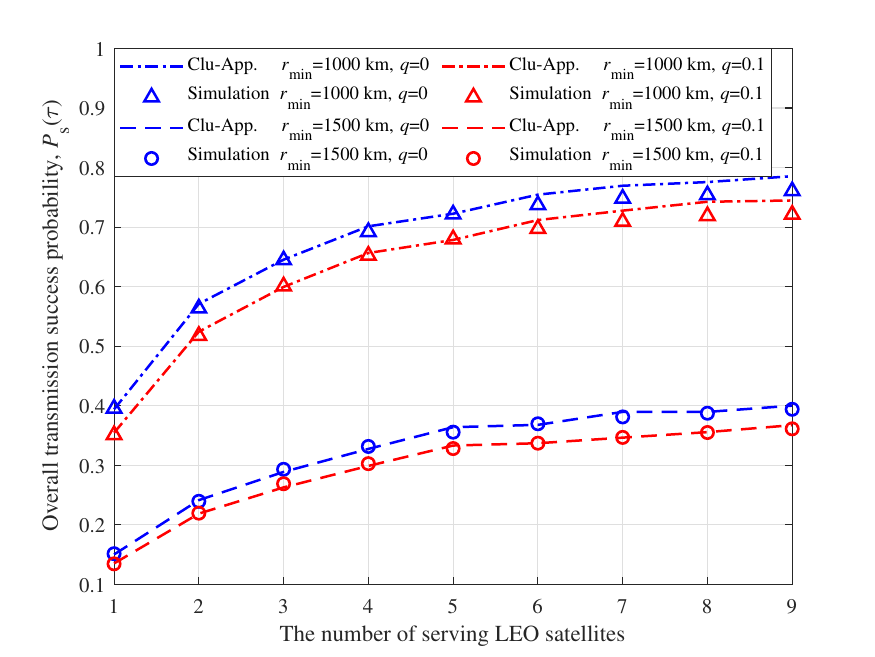}}
		\caption{Validation of the proposed cluster-based approximation approach on the overall transmission success probability with ${\cal K}$=3.}
		\label{cov-app2}
 \end{figure}
\begin{figure}[t]
		\centerline{\includegraphics[scale=0.55]{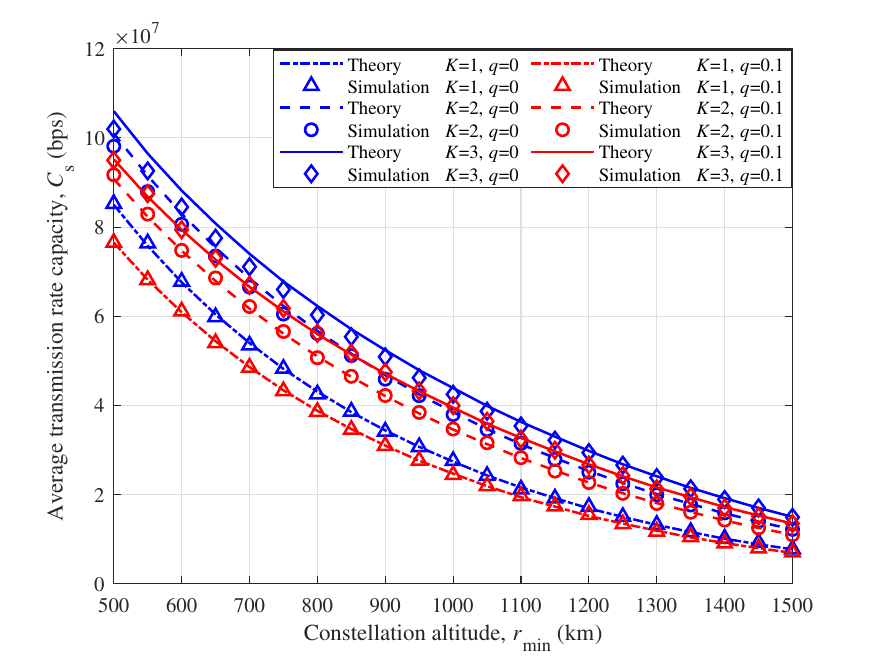}}
		\caption{Effect of the constellation altitude on the average transmission rate capacity.}
		\label{cap-rmin}
                \vspace{-1\baselineskip}
 \end{figure}
\begin{figure}[t]
		\centerline{\includegraphics[scale=0.55]{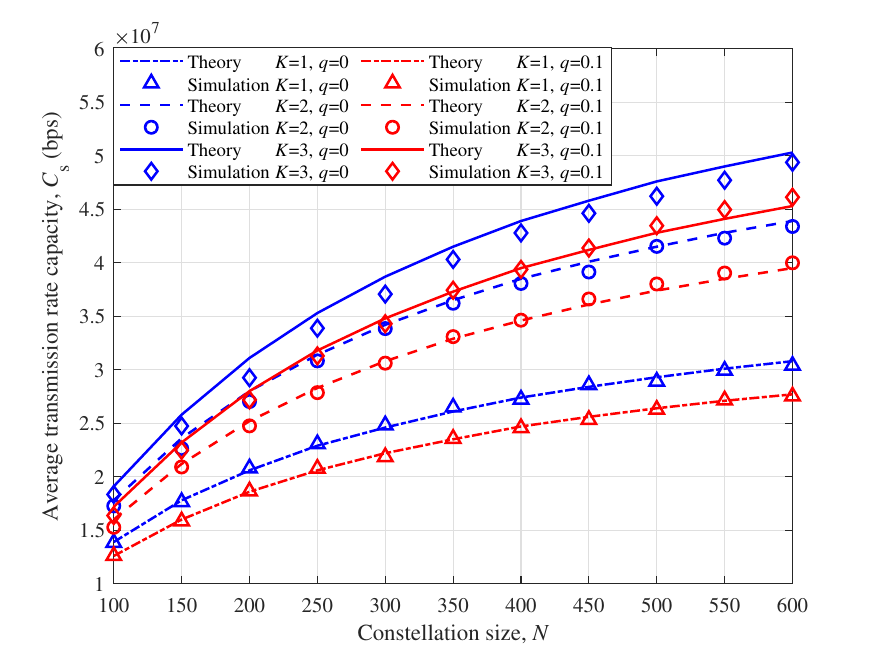}}
		\caption{Effect of the constellation size on the average transmission rate capacity.}
		\label{cap-N}
 \end{figure}
  \begin{figure}[t]
		\centerline{\includegraphics[scale=0.55]{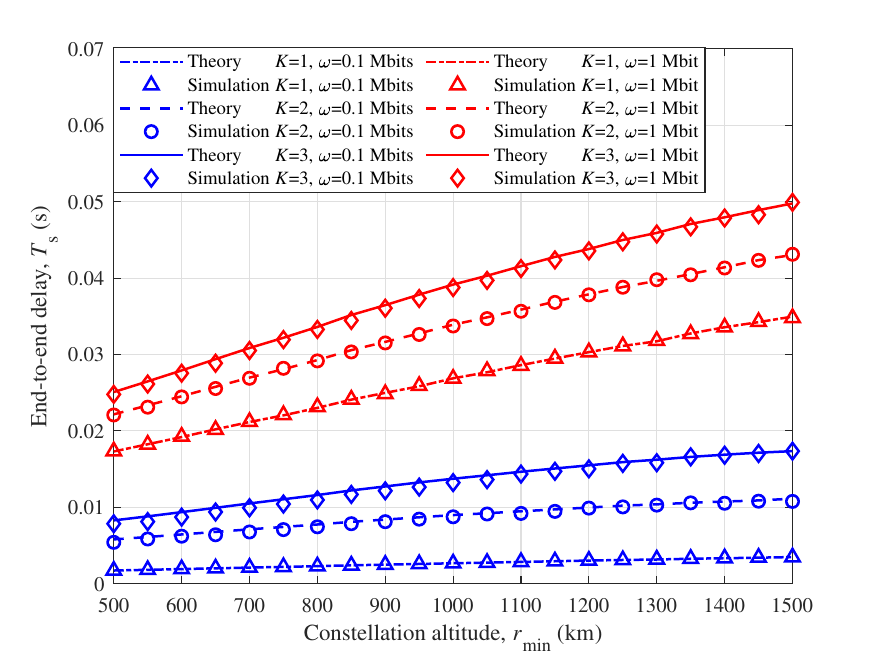}}
		\caption{Effect of the constellation altitude on the end-to-end delay when packet length $\omega$ is 0.1 Mbits and 1 Mbits.}
		\label{delay-rmin}
 \end{figure}
  \begin{figure}[t]
		\centerline{\includegraphics[scale=0.55]{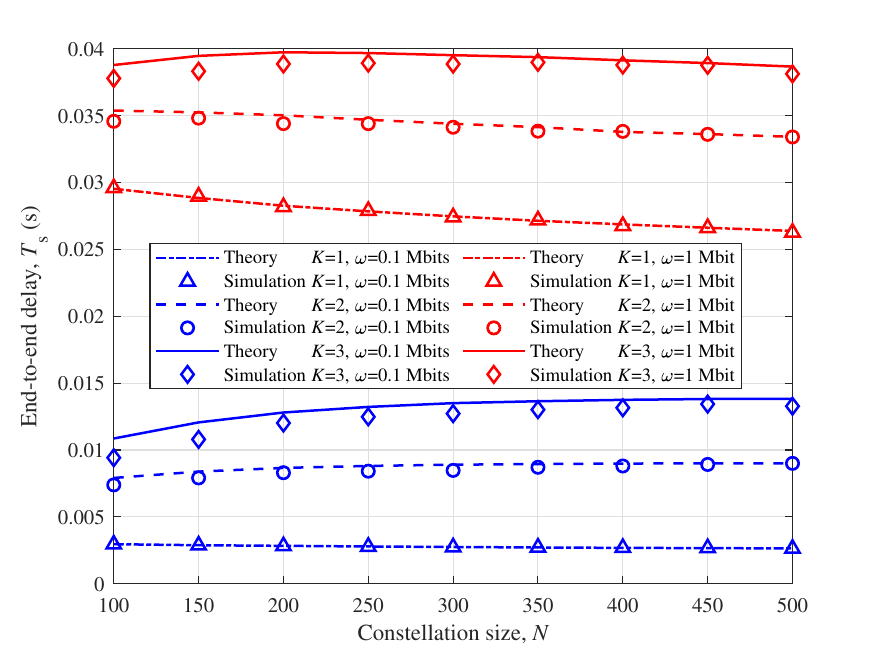}}
		\caption{Effect of the constellation size on the end-to-end delay when packet length $\omega$ is 0.1 Mbits and 1 Mbits.}
		\label{delay-N}
 \end{figure}

Fig. \ref{cov-N} illustrates the relationship between the overall transmission success probability and the constellation size. As the constellation size increases from 100 to 600, a steady rise in the overall transmission success probability is observed. This improvement is attributed to the increased density of LEO satellites, which reduces the average distance between the GU and the serving LEO satellites. The shorter signal propagation distance enhances the quality of both uplink and downlink channels, bringing an improvement in the overall transmission success probability. 

The validation of the derived asymptotic expressions for the overall transmission success probability in two scenarios is presented in Fig. \ref{cov-app1}. Fig. \ref{cov-app1}(a) illustrates the asymptotic expression for \emph{Case 1}, where the predefined SNR threshold ranges from -30 dB to -20 dB, while Fig. \ref{cov-app1}(b) depicts the results for \emph{Case 2}, considering a SNR threshold from -7 dB to 3 dB. The consistency between the analytical results and the simulation outcomes confirms the accuracy of the derived expressions. However, in \emph{Case 1}, the asymptotic expression in (\ref{highSNR}) is valid only in the high-SNR regime, limiting its applicability across the entire SNR range. In contrast, for \emph{Case 2}, results are provided within the SNR range of -7 dB to 3 dB. The results clearly demonstrate the validity of approximating the CDF of SR fading using the Gamma distribution.

The validation of the proposed cluster-based approximation approach with ${\cal K}$=3 is presented in Fig. \ref{cov-app2}. The close alignment between the approximated and simulated results confirms the effectiveness and accuracy of the proposed approach. Furthermore, as the number of serving LEO satellites increases, the overall transmission success probability exhibits a rising trend before gradually stabilizing. This observation suggests that while more serving LEO satellites enhance the coverage performance of the LEO satellite communication network, the enhancement is bounded and does not yield unlimited gains.

Our derivations of the average transmission rate capacity are validated through Monte Carlo simulations, with results presented in Figs. \ref{cap-rmin} and \ref{cap-N}. Specifically, they illustrate how the average transmission rate capacity varies with the constellation altitude and the constellation size, respectively. The results indicate that an increase in constellation altitude leads to a reduction in the average transmission rate capacity. In contrast, as the constellation size grows from 100 to 600, there is a gradual improvement in the average transmission rate capacity. Furthermore, the adoption of the serial communication scheme, e.g., $K$=2 and $K$=3, results in a notable enhancement in the average transmission rate capacity as compared to the single-satellite scenario with $K$=1.

The validation of the end-to-end delay expressions derived in Theorem \ref{thm3} is presented in Figs. \ref{delay-rmin} and \ref{delay-N}. Fig. \ref{delay-rmin} illustrates the impact of constellation altitude on the end-to-end delay when packet length $\omega$ is 0.1 Mbits and 1 Mbits, respectively. It can be seen that as the constellation altitude increases, the end-to-end delay exhibits an upward trend. This is attributed to higher constellation altitudes increasing the distance between LEO satellites and the GU, thereby resulting in longer transmission delays. Furthermore, the end-to-end delay increases with the number of serving LEO satellites due to the additional handover delay introduced in mLEO-CN. The variation of end-to-end delay with respect to the constellation size is depicted in Fig. \ref{delay-N}. It can be observed that the end-to-end delay decreases slightly as the constellation size increases from 100 to 500. This phenomenon arises because with a higher satellite density, the serving satellite in each Monte Carlo realization is more likely to be closer to the GU, thereby shortening the distance between the GU and the serving satellites, and then reducing the uplink and downlink transmission delays through the ${k^{{\rm{th}}}}$ LEO satellite. Consequently, the end-to-end delay decreases. However, when the constellation size is relatively small, a discrepancy exists between the theoretical and simulated values of the PDF describing the distance between an LEO satellite and the GU. This discrepancy results in a deviation in the end-to-end delay.

 \begin{figure}[t] 
		\centerline{\includegraphics[scale=0.55]{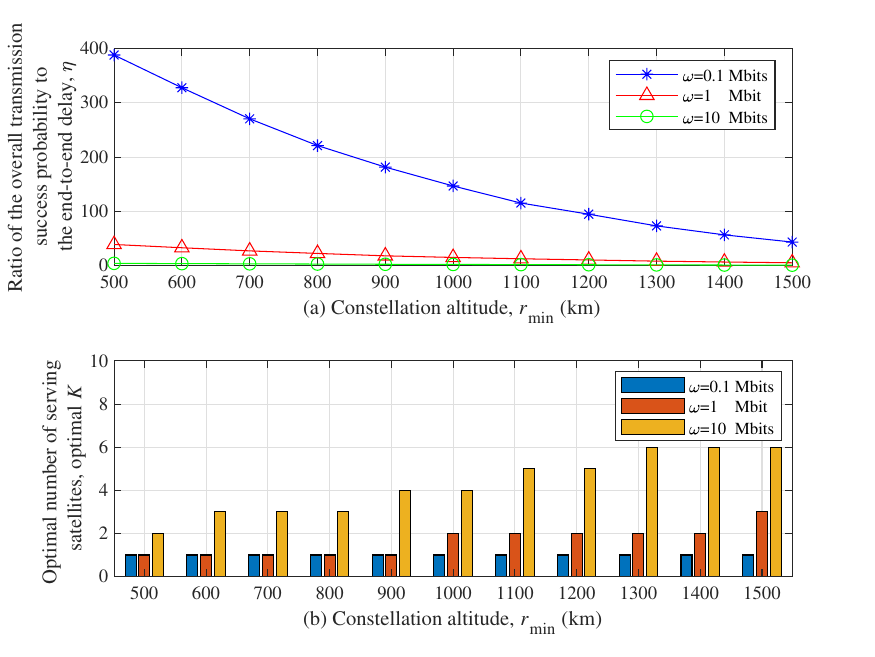}}
		\caption{Optimal number of LEO satellites and maximum overall transmission success probability vary with different constellation altitudes.}
		\label{cov-op}
 \end{figure}
 
Fig. \ref{cov-op} demonstrates the maximum ratio $\eta$ and the optimal number of serving LEO satellites at different constellation altitudes. It is obvious that when the constellation altitude is high, the overall transmission success probability of a single satellite decreases, necessitating a larger number of serving LEO satellites to enhance the overall performance of the mLEO-CN. Conversely, at lower altitudes, a single satellite provides sufficient coverage, making it preferable to reduce the number of serving satellites to reduce the end-to-end delay. For large packet transmissions, the end-to-end delay is primarily determined by transmission delay, with handover delay being negligible. Thus, increasing the number of serving LEO satellites is recommended to enhance network reliability. In contrast, the increased transmission delay makes serial communication supported by multiple LEO satellites less attractive, as frequent switches can degrade network performance despite improving the overall transmission success probability.

\section{Conclusion}
In this paper, we have considered an mLEO-CN where the typical GU is served by multiple LEO satellites serially and developed a theoretical framework to assess its key performance metrics, leveraging the tool of stochastic geometry. The metrics include overall transmission success probability, average transmission rate capacity, and end-to-end delay. To simplify the analytical complexity, we have further derived the asymptotic transmission success probability and proposed a cluster-based approximation for simplification. Based on these analytical results, we have formulated and solved an optimization problem to determine the optimal number of serving satellites that maximizes the trade-off between transmission reliability and latency. Extensive numerical results are presented to evaluate the impact of the key parameters on the end-to-end transmission performance. For the potential follow-up research, first, a key direction is to derive more tractable PDFs of the distance between each LEO satellite and the GU under various spatial distributions, such as PPP and BPP, thereby facilitating more precise system-level performance analysis in scenarios involving multiple satellites. Another important direction is to investigate the network performance of the mLEO-CN under severe co-channel interference, as this would further extend the present analytical framework to more spectrum-reuse-intensive and interference-limited scenarios.


{\appendices
\section*{Appendix}

\subsection{Proof of Theorem 1}
To obtain (\ref{thm1-2}), we derive from the definition of the end-to-end transmission success probability through the ${k^{{\rm{th}}}}$ LEO satellite, which is expressed as
\begin{equation}\label{cov-proof}
\begin{aligned}
    &{P_k}\left( \tau  \right)\buildrel \Delta \over =  \mathbb{P}\left( {{\rm{SNR}}_k^{\rm{u}} > \tau ,{\rm{SNR}}_k^{\rm{d}} > \tau } \right)\\
&= \mathbb{P}\left( {\frac{{{p_{\rm{u}}}G_k^{\rm{u}}L_k\left( {{R_k}} \right){{\left| {H_k^{\rm{u}}} \right|}^2}}}{{\sigma _{k}^2}} > \tau ,\frac{{{p_{{\rm{s}}}}G_k^{\rm{d}}L_k\left( {{R_k}} \right){{\left| {H_k^{\rm{d}}} \right|}^2}}}{{\sigma _{k}^2}} > \tau } \right)\\
 &\underline{\underline {\left( {a} \right)}} \mathbb{P}\left( {\frac{{{p_{\rm{u}}}G_k^{\rm{u}}L_k\left( {{R_k}} \right){{\left| {H_k^{\rm{u}}} \right|}^2}}}{{\sigma _{k}^2}} > \tau } \right)\\
 & \quad \times \mathbb{P}\left( {\frac{{{p_{{\rm{s}}}}G_k^{\rm{d}}L_k\left( {{R_k}} \right){{\left| {H_k^{\rm{d}}} \right|}^2}}}{{\sigma _{k}^2}} > \tau } \right)\\
 &= \mathbb{P}\left( {{\left| {H_k^{\rm{u}}} \right|^2} > \frac{{\tau{\cal G}R_k^\alpha }}{{{p_{{\rm{u}}}}G_k^{\rm{u}}}}} \right)\mathbb{P}\left( {{{\left| {H_k^{\rm{d}}} \right|}^2} > \frac{{\tau{\cal G}R_k^\alpha }}{{{p_{{\rm{s}}}}G_k^{\rm{d}}}}} \right)\\
 &\underline{\underline {\left( {b} \right)}} {\overline F _{{\left| {H_k^{\rm{u}}} \right|^2}}}\left( \frac{{\tau{\cal G}R_k^\alpha }}{{{p_{{\rm{u}}}}G_k^{\rm{u}}}} \right){\overline F _{{\left| {H_k^{\rm{d}}} \right|^2}}}\left( \frac{{\tau{\cal G}R_k^\alpha }}{{{p_{{\rm{s}}}}G_k^{\rm{d}}}} \right)\\
&\underline{\underline {\left( {c} \right)}} \int_{{r_{\min }}}^{{r_{\max }}} {\overline F _{\left| {H_k^{\rm{u}}} \right|^2}}\left( \frac{{\tau{\cal G}r_k^\alpha }}{{{p_{{\rm{u}}}}G_k^{\rm{u}}}} \right) {\overline F _{\left| {H_k^{\rm{d}}} \right|^2}}\left( \frac{{\tau{\cal G}r_k^\alpha }}{{{p_{{\rm{s}}}}G_k^{\rm{d}}}} \right) {f_{{R_k}}}\left( {{r_k}} \right)d{r_k}.
\end{aligned}
\end{equation}
\noindent where $(a)$ is obtained from the independence of the uplink and downlink in the ${k^{{\rm{th}}}}$ channel, $(b)$ employs the complementary CDF (CCDF) of the channel fading gain, and $(c)$ is obtained using the PDF of the distance in Lemma \ref{Lemma-1}. Then, by substituting the expression for the CDF of the SR fading model from step ($c$), (\ref{thm1-2}) is obtained. \hfill $\square$

\subsection{Proof of Corollary 1}
To obtain (\ref{highSNR}), we continue to derive the expression from the last line of (\ref{cov-proof}), given as 
\begin{align*}
  &P_k^\infty \left( \tau  \right)\\
&=\int_{{r_{\min }}}^{{r_{\max }}} \overline F _{{{\left| {H_k^{\rm{u}}} \right|}^2}}^\infty \left( {{{\tau Gr_k^\alpha } \over {{p_{\rm{u}}}G_k^{\rm{u}}}}} \right)\overline F _{{{\left| {H_k^{\rm{d}}} \right|}^2}}^\infty \left( {{{\tau Gr_k^\alpha } \over {{p_{\rm{s}}}G_k^{\rm{d}}}}} \right) {f_{{R_k}}}\left( {{r_k}} \right)d{r_k}\\
&= \int_{{r_{\min }}}^{{r_{\max }}} \left( {1-F _{{{\left| {H_k^{\rm{u}}} \right|}^2}}^\infty \left( {\frac{{\tau{\cal G}r_k^\alpha }}{{{p_{{\rm{u}}}}G_k^{\rm{u}}}}} \right)} \right) \left( {1-F _{{{\left| {H_k^{\rm{d}}} \right|}^2}}^\infty \left( {\frac{{\tau{\cal G}r_k^\alpha }}{{{p_{{\rm{s}}}}G_k^{\rm{d}}}}} \right)} \right)\\
&\quad \times {f_{{R_k}}}\left( {{r_k}} \right)d{r_k}\\
& \overset{\left(a\right)}{\approx} \int_{{r_{\min }}}^{{r_{\max }}} \left(1-{\frac{{\mu \tau {\cal G}r_k^\alpha }}{{{p_{{\rm{u}}}}G_k^{\rm{u}}}}} \right) \left(1- {\frac{{\mu \tau{\cal G}r_k^\alpha }}{{{p_{{\rm{s}}}}G_k^{\rm{d}}}}} \right){f_{{R_k}}}\left( {{r_k}} \right)d{r_k}
\end{align*}
\begin{equation}\label{cor1}
\begin{aligned}
& = \int_{{r_{\min }}}^{{r_{\max }}} \left(1+{\frac{{\mu \tau{\cal G}r_k^\alpha }}{{{p_{{\rm{u}}}}G_k^{\rm{u}}}}}{\frac{{\mu \tau{\cal G}r_k^\alpha }}{{{p_{{\rm{s}}}}G_k^{\rm{d}}}}}-\frac{{\mu \tau {\cal G}r_k^\alpha }}{{{p_{{\rm{u}}}}G_k^{\rm{u}}}}- {\frac{{\mu \tau{\cal G}r_k^\alpha }}{{{p_{{\rm{s}}}}G_k^{\rm{d}}}}}\right)\\
&\quad \times {f_{{R_k}}}\left( {{r_k}} \right)d{r_k},
\end{aligned}
\end{equation}
\noindent where $(a)$ is obtained by employing the approximated CDF of SR fading at high SNR regime in (\ref{appCDF1}). Due to the linear combination of the four items, the expression is expanded as 
\begin{equation}
\begin{aligned}
& \int_{{r_{\min }}}^{{r_{\max }}} {f_{{R_k}}}\left( {{r_k}} \right)d{r_k}+\int_{{r_{\min }}}^{{r_{\max }}}{\frac{{\mu \tau{\cal G}r_k^\alpha }}{{{p_{{\rm{u}}}}G_k^{\rm{u}}}}}{\frac{{\mu \tau{\cal G}r_k^\alpha }}{{{p_{{\rm{s}}}}G_k^{\rm{d}}}}}{f_{{R_k}}}\left( {{r_k}} \right)d{r_k}\\
& -\int_{{r_{\min }}}^{{r_{\max }}} \frac{{\mu \tau{\cal G}r_k^\alpha }}{{{p_{{\rm{u}}}}G_k^{\rm{u}}}}{f_{{R_k}}}\left( {{r_k}} \right)d{r_k} -\int_{{r_{\min }}}^{{r_{\max }}}  {\frac{{\mu \tau{\cal G}r_k^\alpha }}{{{p_{{\rm{s}}}}G_k^{\rm{d}}}}}{f_{{R_k}}}\left( {{r_k}} \right)d{r_k}.
\end{aligned}
\end{equation}
\noindent Based on the property that the integral of the PDF over its entire domain equals one, the first item can be calculated. For the rest of the three items, due to the similarity in the derivation steps, we take the second item as an example to begin the derivation,  which is expressed as
\begin{equation}
\begin{aligned}
& \int_{{r_{\min }}}^{{r_{\max }}} {\frac{{\mu \tau{\cal G}r_k^\alpha }}{{{p_{{\rm{u}}}}G_k^{\rm{u}}}}\frac{{\mu \tau{\cal G}r_k^\alpha }}{{{p_{{\rm{s}}}}G_k^{\rm{d}}}}} {f_{{R_k}}}\left( {{r_k}} \right)d{r_k}\\
& = \frac{{{\mu ^2}\tau^2{{\cal G}^2}}}{{{p_{{\rm{u}}}}G_k^{\rm{u}}{p_{{\rm{s}}}}G_k^{\rm{d}}\Xi }}\int_{{r_{\min }}}^{{r_{\max }}} {r_k^{2\alpha }\sum\limits_{l = 0}^{k - 1} {\sum\limits_{n = 0}^{N - k} {\frac{{\psi }}{\Theta }r_k^{2n + 2l + 1}} } } d{r_k}\\
& = \frac{{{\mu ^2}\tau^2{{\cal G}^2}}}{{{p_{{\rm{u}}}}G_k^{\rm{u}}{p_{{\rm{s}}}}G_k^{\rm{d}}\Xi }}\sum\limits_{l = 0}^{k - 1} {\sum\limits_{n = 0}^{N - k} {\frac{\psi }{\Theta }\int_{{r_{\min }}}^{{r_{\max }}} {r_k^{2\alpha  + 2n + 2l + 1}} d{r_k}} } \\
&= \frac{{{\mu ^2}\tau^2{{\cal G}^2}}}{{{p_{{\rm{u}}}}G_k^{\rm{u}}{p_{{\rm{s}}}} G_k^{\rm{d}}\Xi }}\sum\limits_{l = 0}^{k - 1} {\sum\limits_{n = 0}^{N - k} {\frac{\psi }{\Theta }\left( {\frac{{r_{\max }^{2\alpha  + \Delta} - r_{\min }^{2\alpha  + \Delta }}}{{2\alpha  + \Delta}}} \right)} }.
\end{aligned}
\end{equation}
\noindent Following similar steps, we can derive similar closed-form expressions for other items, and then $(\ref{highSNR})$ is obtained. \hfill $\square$

\subsection{Proof of Corollary 2}
To obtain (\ref{covapp3}), we start from the last line of (\ref{cov-proof}), given as
\begin{align*}
    &{P_k}\left( \tau  \right)\buildrel \Delta \over =  \mathbb{P}\left( {{\rm{SNR}}_k^{\rm{u}} > \tau ,{\rm{SNR}}_k^{\rm{d}} > \tau } \right)\\
&= \mathbb{P}\left( {\frac{{{p_{\rm{u}}}G_k^{\rm{u}}L_k\left( {{R_k}} \right){{\left| {H_k^{\rm{u}}} \right|}^2}}}{{\sigma _{k}^2}} > \tau ,\frac{{{p_{{\rm{s}}}}G_k^{\rm{d}}L_k\left( {{R_k}} \right){{\left| {H_k^{\rm{d}}} \right|}^2}}}{{\sigma _{k}^2}} > \tau } \right)\\
& \overset{(a)}{\approx} \int_{{r_{\min }}}^{{r_{\max }}} {\left(1- {\frac{{\gamma  \left( {{\alpha _s},\frac{1}{{{\beta _s}}}\frac{{\tau{\cal G}r_k^\alpha }}{{{p_{{\rm{u}}}}G_k^{\rm{u}}}}} \right)}}{{\Gamma \left( {{\alpha _s}} \right)}}} \right)\left( {1-\frac{{\gamma \left( {{\alpha _s},\frac{1}{{{\beta _s}}}\frac{{\tau{\cal G}r_k^\alpha }}{{{p_{{\rm{s}}}}G_k^{\rm{d}}}}} \right)}}{{\Gamma \left( {{\alpha _s}} \right)}} } \right) }\\
&\quad \times {f_{{R_k}}}\left( {{r_k}} \right)d{r_k}\\
& \underline{\underline {\left( b \right)}}  \int_{{r_{\min }}}^{{r_{\max }}} {\frac{{\Gamma \left( {{\alpha _s},\frac{1}{{{\beta _s}}}\frac{{\tau{\cal G}r_k^\alpha }}{{{p_{{\rm{u}}}}G_k^{\rm{u}}}}} \right)}}{{\Gamma \left( {{\alpha _s}} \right)}}\frac{{\Gamma \left( {{\alpha _s},\frac{1}{{{\beta _s}}}\frac{{\tau{\cal G}r_k^\alpha }}{{{p_{{\rm{s}}}}G_k^{\rm{d}}}}} \right)}}{{\Gamma \left( {{\alpha _s}} \right)}}} {f_{{R_k}}}\left( {{r_k}} \right)d{r_k}\\
& \underline{\underline {\left( c \right)}}   \int_{{r_{\min }}}^{{r_{\max }}} {\exp \left( { - \frac{{\tau{\cal G}r_k^\alpha }}{{{p_{{\rm{u}}}}G_k^{\rm{u}}}}} \right)} \sum\limits_{{t_1} = 0}^{\widetilde {{\alpha _s}} - 1} {\frac{{{{\left( {\frac{{\tau{\cal G}r_k^\alpha }}{{{p_{{\rm{u}}}}G_k^{\rm{u}}}}} \right)}^{{t_1}}}}}{{\Gamma \left( { {t_1} + 1} \right)}}}\\
&\quad \times{ \exp \left( { - \frac{{\tau{\cal G}r_k^\alpha }}{{{p_{{\rm{s}}}}G_k^{\rm{d}}}}} \right)} \sum\limits_{{t_2} = 0}^{\widetilde {{\alpha _s}} - 1} {\frac{{{{\left( {\frac{{\tau{\cal G}r_k^\alpha }}{{{p_{{\rm{s}}}}G_k^{\rm{d}}}}} \right)}^{{t_2}}}}}{{\Gamma \left( { {t_2} + 1} \right) }}} {f_{{R_k}}}\left( {{r_k}} \right) d{r_k}\\
& = \frac{1}{\Xi }\sum\limits_{l = 0}^{k - 1} {\sum\limits_{n = 0}^{N - k} {\sum\limits_{{t_1} = 0}^{\widetilde {{\alpha _s}} - 1} {\sum\limits_{{t_2} = 0}^{\widetilde {{\alpha _s}} - 1} {\frac{\psi }{\Theta }\frac{{{{\left( {\frac{\tau{\cal G}}{{{p_{{\rm{u}}}}G_k^{\rm{u}}}}} \right)}^{{t_1}}}{{\left( {\frac{\tau{\cal G}}{{{p_{{\rm{s}}}}G_k^{\rm{d}}}}} \right)}^{{t_2}}}}}{{\Gamma \left( { {t_1} + 1} \right)\Gamma \left( { {t_2} + 1} \right)}}}  } } }
\end{align*}
 \begin{equation}
    \begin{aligned}\label{cor3}
&\quad \times \int_{{r_{\min }}}^{{r_{\max }}} {r_k^{\left( {{t_1} + {t_2}} \right)\alpha  + 2n + 2l + 1}\exp \left( { - {\rm A}r_k^\alpha } \right)} d{r_k}\\
&=\frac{1}{\Xi }\sum\limits_{l = 0}^{k - 1} {\sum\limits_{n = 0}^{N - k} {\sum\limits_{{t_1} = 0}^{\widetilde {{\alpha _s}} - 1} {\sum\limits_{{t_2} = 0}^{\widetilde {{\alpha _s}} - 1} {\frac{\psi }{\Theta }\frac{{{{\left( {\frac{\tau{\cal G}}{{{p_{{\rm{u}}}}G_k^{\rm{u}}}}} \right)}^{{t_1}}}{{\left( {\frac{\tau{\cal G}}{{{p_{{\rm{s}}}}G_k^{\rm{d}}}}} \right)}^{{t_2}}}}}{{\alpha \Gamma \left( { {t_1} + 1} \right)\Gamma \left( { {t_2} + 1} \right)}}} } } }\\
&\quad \times \left( \frac{{r_{{\rm{min}}}^{\rm{B}}\Gamma \left( {\frac{{\rm{B}}}{\alpha }{\rm{A}}r_{{\rm{min}}}^\alpha } \right)}}{{{{\left( {{\rm{A}}r_{{\rm{min}}}^\alpha } \right)}^{\frac{{\rm{B}}}{\alpha }}}}}-\frac{{r_{{\rm{max}}}^{\rm{B}}\Gamma \left( {\frac{{\rm{B}}}{\alpha }{\rm{A}}r_{{\rm{max}}}^\alpha } \right)}}{{{{\left( {{\rm{A}}r_{{\rm{max}}}^\alpha } \right)}^{\frac{{\rm{B}}}{\alpha }}}}} \right),
\end{aligned}
\end{equation}
\noindent where $(a)$ is obtained using the approximated CDF of the SR fading in $(\ref{cdfappsr3})$, $(b)$ is from the relationship between the upper and lower incomplete gamma functions, i.e., $\frac{{\gamma (s,x)}}{{{\rm{\Gamma }}(s)}}{\rm{ + }}\frac{{{\rm{\Gamma }}(s,x)}}{{{\rm{\Gamma }}(s)}}=1$. In $(c)$, since $\widetilde {{\alpha _s}}$ is an integer, $(c)$ is obtained from
 ${\rm{\Gamma }}\left( {n,x} \right) = \left( {n - 1} \right)!{{\rm{e}}^{ - x}}\sum\limits_{k = 0}^{n - 1} {\frac{{{x^k}}}{{k!}}}$ \cite{alzenad2019coverage}.  \hfill $\square$
 
 \subsection{Proof of Theorem 2}
To obtain (\ref{capi}), we start with the definition of the end-to-end transmit rate capacity through the ${k^{{\rm{th}}}}$ LEO satellite, which is shown in (\ref{ProofofCi}), at the top of the next page, where $(a)$ is obtained by transforming the condition of both the uplink from the GU to the ${k^{{\rm{th}}}}$ LEO satellite and the downlink from the ${k^{{\rm{th}}}}$ LEO satellite to the GU to indicator functions; $(b)$ employs the PDF of the SR fading model, which is given in $(\ref{pdfofSR})$; $(c)$ is obtained by averaging ${R_k}$ using the PDF in Lemma \ref{Lemma-1}. \hfill $\square$
\begin{figure*}[t]
\begin{equation}\label{ProofofCi}
 \begin{aligned}
{C_k}& \buildrel \Delta \over = B{\log _2}\left( {1 + \min \left( {{\rm{SNR}}_k^{\rm{u}},{\rm{SNR}}_k^{\rm{d}}} \right)|\left( {{\rm{SNR}}_k^{\rm{u}} > \tau ,{\rm{SNR}}_k^{\rm{u}} > \tau } \right)} \right)\\
&= B{\log _2}\left( {1 + \min \left( {\frac{{{p_{{\rm{u}}}}G_k^{\rm{u}}{{\left| {H_k^{\rm{u}}} \right|}^2}R_k^{ - \alpha }}}{{\cal G}},\frac{{{p_{{\rm{s}}}}G_k^{\rm{d}}{{\left| {H_k^{\rm{d}}} \right|}^2}R_k^{ - \alpha }}}{{\cal G}}} \right) } \right.  \left. {|\left( {\frac{{{p_{{\rm{u}}}}G_k^{\rm{u}}{{\left| {H_k^{\rm{u}}} \right|}^2}R_k^{ - \alpha }}}{{\cal G}} > \tau ,\frac{{{p_{{\rm{s}}}}G_k^{\rm{d}}{{\left| {H_k^{\rm{d}}} \right|}^2}R_k^{ - \alpha }}}{{\cal G}} > \tau } \right)} \right)\\
&\underline{\underline {\left( a \right)}} B{\log _2}\left( {1 + \min \left( {\frac{{{p_{{\rm{u}}}}G_k^{\rm{u}}{{\left| {H_k^{\rm{u}}} \right|}^2}R_k^{ - \alpha }}}{{\cal G}},\frac{{{p_{{\rm{s}}}}G_k^{\rm{d}}{{\left| {H_k^{\rm{d}}} \right|}^2}R_k^{ - \alpha }}}{{\cal G}}} \right)} \right) \mathbbm{1} \left( {\frac{{{p_{{\rm{u}}}}G_k^{\rm{u}}{{\left| {H_k^{\rm{u}}} \right|}^2}R_k^{ - \alpha }}}{{\cal G}}} >\tau\right) \mathbbm{1} \left( {\frac{{{p_{{\rm{s}}}}G_k^{\rm{d}}{{\left| {H_k^{\rm{d}}} \right|}^2}R_k^{ - \alpha }}}{{\cal G}}}>\tau \right)\\
&\underline{\underline {\left( b \right)}} B\int_0^\infty  {\int_0^\infty  {{{\log }_2}\left( {1 + \min \left( {\frac{{{p_{{\rm{u}}}}G_k^{\rm{u}}{{\left| {h_k^{\rm{u}}} \right|}^2}R_k^{ - \alpha }}}{{\cal G}},\frac{{{p_{{\rm{s}}}}G_k^{\rm{d}}{{\left| {h_k^{\rm{d}}} \right|}^2}R_k^{ - \alpha }}}{{\cal G}}} \right)} \right)} }  \mathbbm{1} \left( {\frac{{{p_{{\rm{u}}}}G_k^{\rm{u}}{{\left| {h_k^{\rm{u}}} \right|}^2}R_k^{ - \alpha }}}{{\cal G}}}>\tau \right) \mathbbm{1} \left( {\frac{{{p_{{\rm{s}}}}G_k^{\rm{d}}{{\left| {h_k^{\rm{d}}} \right|}^2}R_k^{ - \alpha }}}{{\cal G}}}>\tau \right)\\
&\quad \times {f_{\left| {H_k^{\rm{u}}} \right|^2}}\left( {{{h_k^{\rm{u}}}}} \right){f_{\left| {H_k^{\rm{u}}} \right|^2}}\left( {{h_k^{\rm{d}}}} \right)d{{h_k^{\rm{u}}}}d{{h_k^{\rm{d}}}}\\
&\underline{\underline {\left( c \right)}} B\int_{{r_{\min }}}^{{r_{\max }}} {\int_0^\infty  {\int_0^\infty  {{{\log }_2}\left( {1 + \min \left( {\frac{{{p_{{\rm{u}}}}G_k^{\rm{u}}{{\left| {h_k^{\rm{u}}} \right|}^2}r_k^{ - \alpha }}}{{\cal G}},\frac{{{p_{{\rm{s}}}}G_k^{\rm{d}}{{\left| {h_k^{\rm{d}}} \right|}^2}r_k^{ - \alpha }}}{{\cal G}}} \right)} \right)} } }  \mathbbm{1} \left( {\frac{{{p_{{\rm{u}}}}G_k^{\rm{u}}{{\left| {h_k^{\rm{u}}} \right|}^2}r_k^{ - \alpha }}}{{\cal G}} }>\tau \right) \\
&\quad \times \mathbbm{1}\left( {\frac{{{p_{{\rm{s}}}}G_k^{\rm{d}}{{\left| {h_k^{\rm{d}}} \right|}^2}r_k^{ - \alpha }}}{{\cal G}} }>\tau \right) {f_{\left| {H_k^{\rm{u}}} \right|^2}}\left( {{{h_k^{\rm{u}}}}} \right){f_{\left| {H_k^{\rm{u}}} \right|^2}}\left( {{h_k^{\rm{d}}}} \right)d{{h_k^{\rm{u}}}}d{{h_k^{\rm{d}}}}f_{R_k}\left( {{r_k}} \right)d{r_k}.
\end{aligned}
\end{equation}\hrulefill
\end{figure*}

 \subsection{Proof of Theorem 3}
In (\ref{delayk}), there are three items, that is $T_k^{\rm{u}}$, $T_k^{\rm{d}}$, and $\left( {k - 1} \right) {T_{{\rm{ho}}}}$. Due to their similar derivation steps, we take the first item as an example to start the derivation, given as 
\begin{align*}
T_k^{\rm{u}}& \buildrel \Delta \over =  \frac{\omega}{{B{{\log }_2}\left( {1 + {\rm{SNR}}_k^{\rm{u}}} \right)}}\\
 &=\left( { \frac{\omega}{{B{{\log }_2}\left( {1 + \frac{{{p_{{\rm{u}}}}G_k^{\rm{u}}{{\left| {H_k^{\rm{u}}} \right|}^2}R_k^{ - \alpha }}}{{\cal G}}} \right)}}} \right.
\end{align*}
\begin{equation}\label{proofofdelay}
    \begin{aligned}
    &\left. {| \frac{{{p_{{\rm{u}}}}G_k^{\rm{u}}{{\left| {H_k^{\rm{u}}} \right|}^2}R_k^{ - \alpha }}}{{\cal G}}> \tau, \frac{{{p_{{\rm{s}}}}G_k^{\rm{d}}{{\left| {H_k^{\rm{d}}} \right|}^2}R_k^{ - \alpha }}}{{\cal G}}> \tau } \right)\\
     & \underline{\underline {\left( {a} \right)}}  \int_0^\infty \int_0^\infty {\frac{\omega}{{B{{\log }_2}\left( {1 + \frac{{{p_{{\rm{u}}}}G_k^{\rm{u}}{{\left| {h_k^{\rm{u}}} \right|}^2}R_k^{ - \alpha }}}{{\cal G}}} \right)}}}\\
&\quad \times \mathbbm{1}\left( {\frac{{{p_{{\rm{u}}}}G_k^{\rm{u}}{{\left| {h_k^{\rm{u}}} \right|}^2}R_k^{ - \alpha }}}{{\cal G}}> \tau} \right)\mathbbm{1}\left( {\frac{{{p_{{\rm{s}}}}G_k^{\rm{d}}{{\left| {h_k^{\rm{d}}} \right|}^2}R_k^{ - \alpha }}}{{\cal G}}> \tau} \right)\\
&\quad \times {f_{\left| {H_k^{\rm{u}}} \right|^2}}\left( h_k^{\rm{u}} \right) {f_{\left| {H_k^{\rm{d}}} \right|^2}}\left( h_k^{\rm{d}} \right)d{h_k^{\rm{u}}}d{h_k^{\rm{d}}}\\
& \underline{\underline {\left( {b} \right)}} \int_{{r_{\min }}}^{{r_{\max }}} {\int_0^\infty \int_0^\infty  {\frac{\omega}{{B{{\log }_2}\left( {1 + \frac{{{p_{{\rm{u}}}}G_k^{\rm{u}}{{\left| {h_k^{\rm{u}}} \right|}^2}r_k^{ - \alpha }}}{{\cal G}}} \right)}} }}\\
&\quad \times \mathbbm{1}\left( {\frac{{{p_{{\rm{u}}}}G_k^{\rm{u}}{{\left| {h_k^{\rm{u}}} \right|}^2}r_k^{ - \alpha }}}{{\cal G}}> \tau} \right)\mathbbm{1}\left( {\frac{{{p_{{\rm{s}}}}G_k^{\rm{d}}{{\left| {h_k^{\rm{d}}} \right|}^2}r_k^{ - \alpha }}}{{\cal G}}> \tau} \right)\\
&\quad \times {{f_{\left| {H_k^{\rm{u}}} \right|^2}}\left( h_k^{\rm{u}} \right)} {f_{\left| {H_k^{\rm{d}}} \right|^2}}\left( h_k^{\rm{d}} \right){f_{{R_k}}}\left( {{r_k}} \right)d{h_k^{\rm{u}}}d{h_k^{\rm{d}}}d{r_k}, 
\end{aligned}
\end{equation}
\noindent where ($a$) is obtained by applying the PDF of the SR fading model and converting the conditional probability into indicator functions; ($b$) is from the PDF of $R_k$ in Lemma 1. Similarly, referring to the derivation above to average $T_k^{\rm{d}}$ and $\left( {k - 1} \right) {T_{{\rm{ho}}}}$, we can obtain (\ref{delayTk}). \hfill $\square$

 \subsection{Proof of Proposition 2}
 \color{black}{
To clearly characterize the dependence of the objective function on the number of serving satellites $K$, we denote $\eta$ by $F(K)$, expressed as $F(K) \triangleq \eta$.

For any $K{\rm{ > 2}}$, define the following quantities for notational simplicity
\begin{equation}
a_k \triangleq \bar q_k P_k(\tau), \qquad b_k \triangleq P_k(\tau).
\end{equation}
Since $P_k(\tau)$ are $\bar q_k{P_k}(\tau)$ are non-increasing with respect to $k$, $T_k$ is non-decreasing with respect to $k$, we have
\begin{equation}
\begin{aligned}
&a_1 \ge a_2 \ge \cdots \ge a_k, \quad
b_1 \ge b_2 \ge \cdots \ge b_k,\\
&T_1 \le T_2 \le \cdots \le T_k.
\end{aligned}
\end{equation}

Define
\begin{equation}
U_K \triangleq \prod_{k=1}^{K}(1-a_k), \qquad
V_K \triangleq \prod_{k=1}^{K}(1-b_k),
\end{equation}
\begin{equation}
Y_K \triangleq (1-U_K)(1-V_K),
\end{equation}
\begin{equation}
W_K \triangleq T_1+\sum_{k=2}^{K}T_kV_{k-1},
\end{equation}
Then, the objective function can be rewritten as
\begin{equation}
F(K)=\frac{Y_K}{W_K}.
\end{equation}

Since $K$ is a discrete variable, unimodality is established by studying the sign of the forward difference $F(K+1)-F(K)$.

From the definitions of $U_K$ and $V_K$, we obtain
\begin{equation}
U_{K+1}=U_K(1-a_{K+1}), \qquad
V_{K+1}=V_K(1-b_{K+1}).
\end{equation}
Hence,
\begin{align}
\Delta Y_K
&\triangleq Y_{K+1}-Y_K \notag\\
&=(1-U_K)b_{K+1}V_K+(1-V_K)a_{K+1}U_K \notag\\
&\quad +a_{K+1}b_{K+1}U_KV_K,
\end{align}
and
\begin{equation}
\Delta W_K \triangleq W_{K+1}-W_K=T_{K+1}V_K.
\end{equation}
Therefore,
\begin{align}
F(K+1)-F(K)
&=\frac{Y_K+\Delta Y_K}{W_K+\Delta W_K}-\frac{Y_K}{W_K} \notag\\
&=\frac{W_K\Delta Y_K-Y_K\Delta W_K}{W_K(W_K+\Delta W_K)}.
\end{align}
Since the denominator is positive, the sign of $F(K+1)-F(K)$ is determined by
\begin{equation}
W_K\Delta Y_K-Y_K\Delta W_K.
\end{equation}

Define the marginal gain-to-cost ratio
\begin{equation}
M_K \triangleq \frac{\Delta Y_K}{\Delta W_K}.
\end{equation}
Then,
\begin{equation}
F(K+1)-F(K)
=\frac{\Delta W_K}{W_K+\Delta W_K}\bigl(M_K-F(K)\bigr).
\end{equation}
Thus,
\begin{equation}
F(K+1)\ge F(K)\iff M_K\ge F(K),
\end{equation}
and
\begin{equation}
F(K+1)\le F(K)\iff M_K\le F(K).
\end{equation}

Therefore, the monotonic behavior of $F(K)$ is fully characterized by the comparison between $M_K$ and $F(K)$.

Using the expressions of $\Delta Y_K$ and $\Delta W_K$, we have
\begin{align}
M_K
&=\frac{(1-U_K)b_{K+1}}{T_{K+1}}
+\frac{(1-V_K)a_{K+1}}{T_{K+1}}\frac{U_K}{V_K} \notag\\
&\quad+\frac{a_{K+1}b_{K+1}U_K}{T_{K+1}}.
\end{align}

Since $a_{K+1}$ and $b_{K+1}$ decrease with $K$, whereas $T_{K+1}$ increases with $K$. From the physical interpretation, this means that when more serving satellites are included, the additional performance gain brought by the ${\left( {k + 1} \right)^{{\rm{th}}}}$ serving satellite becomes smaller, while the corresponding delay cost becomes larger. As a result, the marginal gain-to-cost ratio $M_K$ is non-increasing with respect to $K$, i.e.,
\begin{equation}
M_{K+1}\le M_K.
\end{equation}

Suppose that for some ${K^{'}}$,
\begin{equation}
F(K^{'}+1)\le F(K^{'}).
\end{equation}
Then, from
\begin{equation}
F(K+1)-F(K)
=
\frac{\Delta W_K}{W_K+\Delta W_K}\bigl(M_K-F(K)\bigr),
\end{equation}
we obtain
\begin{equation}
M_{K^{'}}\le F(K^{'}).
\end{equation}

On the other hand,
\begin{equation}
\begin{aligned}
F(K^{'}+1)
=&
\frac{W_{K^{'}}}{W_{K^{'}}+\Delta W_{K^{'}}}F(K^{'})\\
&+
\frac{\Delta W_{K^{'}}}{W_{K^{'}}+\Delta W_{K^{'}}}M_{K^{'}},
\end{aligned}
\end{equation}
which shows that $F(K^{'}+1)$ is a weighted average of $F(K^{'})$ and $M_{K^{'}}$. Therefore,
\begin{equation}
M_{K^{'}}\le F(K^{'}+1)\le F(K^{'}).
\end{equation}

Since $M_K$ is non-increasing, we further have
\begin{equation}
M_{K^{'}+1}\le M_{K^{'}}\le F(K^{'}+1).
\end{equation}
Hence,
\begin{equation}
\begin{aligned}
&F(K^{'}+2)-F(K^{'}+1)
=\\
&\frac{\Delta W_{K^{'}+1}}{W_{K^{'}+1}+\Delta W_{K^{'}+1}}\bigl(M_{K^{'}+1}-F(K^{'}+1)\bigr)
\le 0.
\end{aligned}
\end{equation}

Repeating the same argument recursively yields
\begin{equation}
F(K+1)\le F(K), \qquad \forall K\ge K^{'}.
\end{equation}

Therefore, once $F(K)$ starts to decrease, it will keep decreasing afterwards. This implies that the sequence $\{F(K)\}$ can experience at most one transition from increasing to decreasing. Hence, $F(K)$ is unimodal with respect to $K$ when $K \ge 2$. 

Finally, although the case $K=1$ is written separately for clarity, it is in fact the boundary case of the same objective function. Specifically, substituting $K=1$ into the general expression yields $
F(1)=\frac{\bar q_1\,P_1^2(\tau)}{T_1}$,
which is exactly the separately stated form. Therefore, the piecewise presentation introduces no inconsistency, and the unimodality result applies to the entire discrete domain $K\ge 1$.}  \hfill $\square$
 
This concludes the proof.
}

\bibliographystyle{IEEEtran}
\bibliography{refwork2} 
\newpage
    
\begin{IEEEbiography}
[{\includegraphics[width=1in,height=1.25in,clip,keepaspectratio]{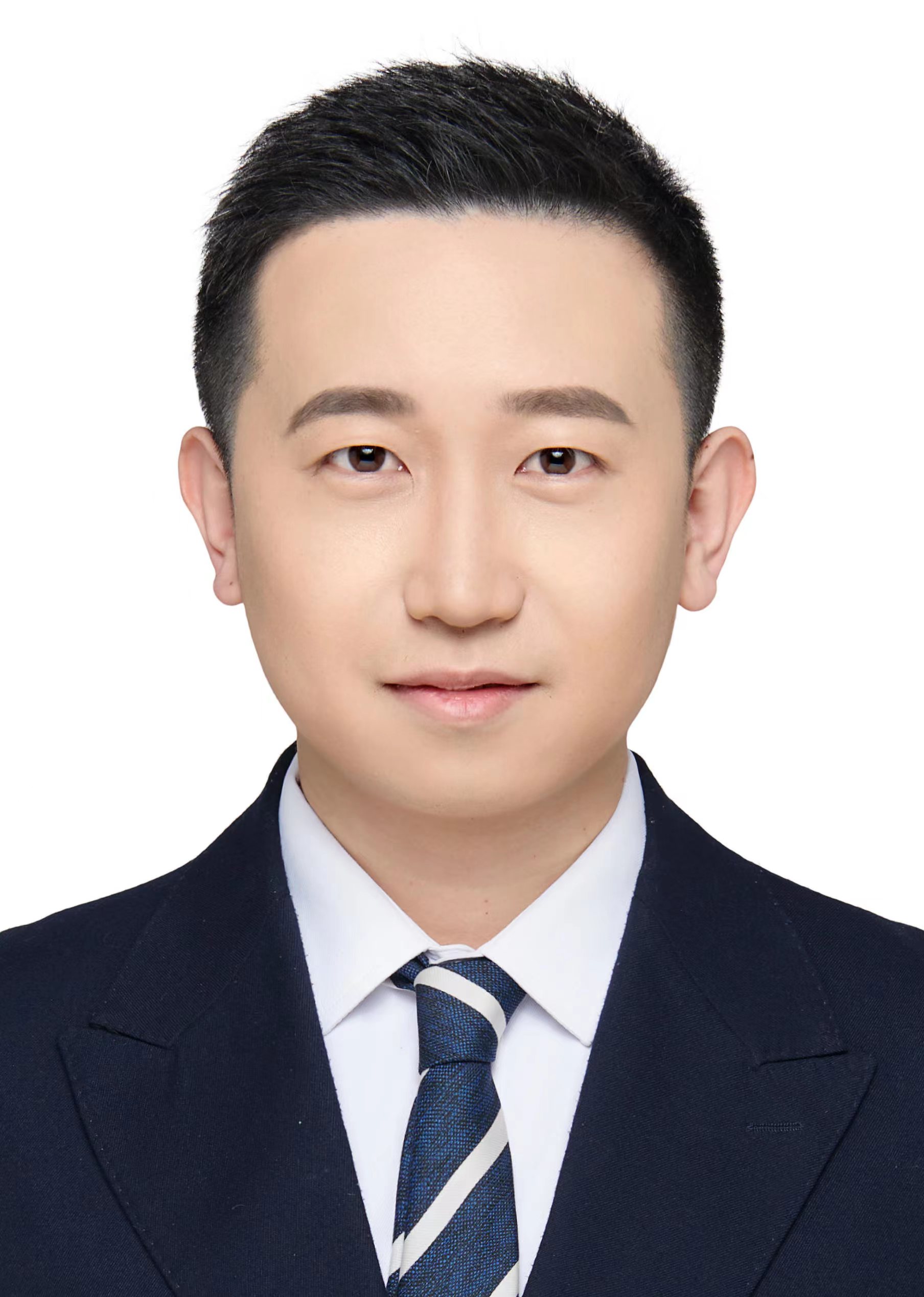}}]{Xu Hu}
(Graduate Student Member, IEEE) received the B.S. degree from Dalian Polytechnic University, Dalian, China, in 2017, and received the M.S. degree from Dalian Maritime University, Dalian, China, in 2020. He is currently pursuing the Ph.D. degree in Information Science and Technology College at Dalian Maritime University, Dalian, China. His research interests include maritime communication networks, LEO satellite communications, network performance analysis, and stochastic geometry.
\end{IEEEbiography}

\begin{IEEEbiography}
[{\includegraphics[width=1in,height=1.25in,clip,keepaspectratio]{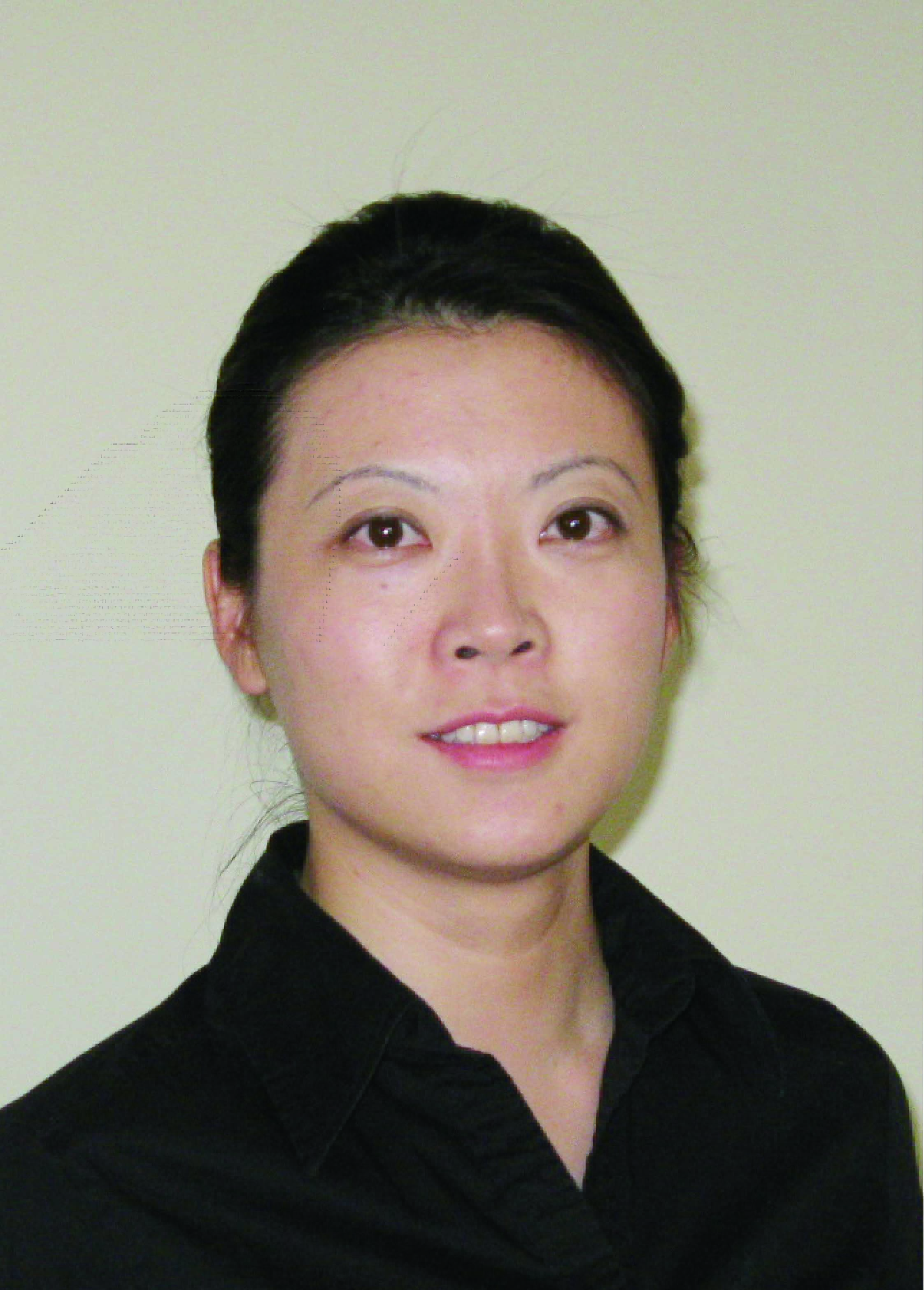}}]{Bin Lin}
(Senior Member, IEEE) received the B.S. and M.S. degrees from Dalian Maritime University, Dalian, China, in 1999 and 2003, respectively and the Ph.D. degree from the Broadband Communications Research Group, Department of Electrical and Computer Engineering, University of Waterloo, Waterloo, ON, Canada, in 2009. She is currently a Full Professor and the Dean of Communication Engineering with the College of Information Science and Technology, Dalian Maritime University. She has been a Visiting Scholar with George Washington University, Washington, DC, USA, from 2015 to 2016. Her current research interests include wireless communications, network dimensioning and optimization, resource allocation, artificial intelligence, maritime communication networks, and Internet of Things. Prof. Lin is an associate editor of IEEE Transactions on Vehicular Technology and IEEE Internet of Things Journal.
\end{IEEEbiography}

\begin{IEEEbiography}
[{\includegraphics[width=1in,height=1.25in,clip,keepaspectratio]{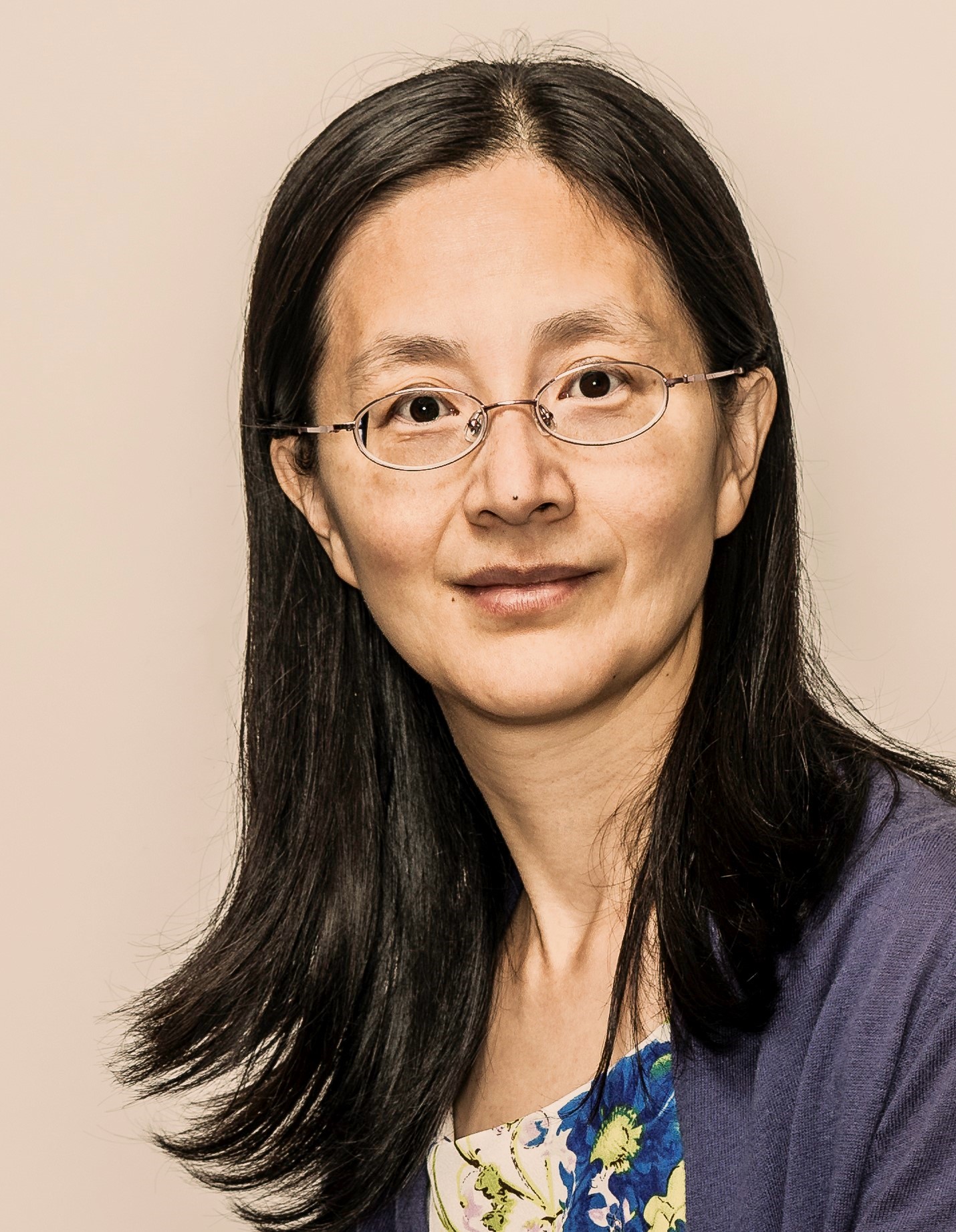}}]{Ping Wang} (Fellow, IEEE) is a Professor at the Department of Electrical Engineering and Computer Science, York University, and a Tier 2 York Research Chair. Prior to that, she was with Nanyang Technological University, Singapore, from 2008 to 2018.  Her research interests are mainly in radio resource allocation, network design, performance analysis and optimization for wireless communication networks, mobile cloud computing and the Internet of Things. Her recent works focus on integrating Artificial Intelligence (AI) techniques into communications networks. Her scholarly works have been widely disseminated through top-ranked IEEE journals/conferences, received 28,000+ citations, and received the Best Paper Awards from IEEE prestigious conference WCNC in 2012, 2020 and 2022, from IEEE Communication Society: Green Communications ${\rm{\& }}$ Computing Technical Committee in 2018, from IEEE flagship conference ICC in 2007. She has been serving as associate editor-in-chief for IEEE Communications Surveys ${\rm{\& }}$ Tutorials and an editor for several reputed journals, including IEEE Transactions on Wireless Communications. She is an IEEE Fellow and a \emph{Distinguished Lecturer of the IEEE Vehicular Technology Society}.
\end{IEEEbiography}

\begin{IEEEbiography}
[{\includegraphics[width=1in,height=1.25in,clip,keepaspectratio]{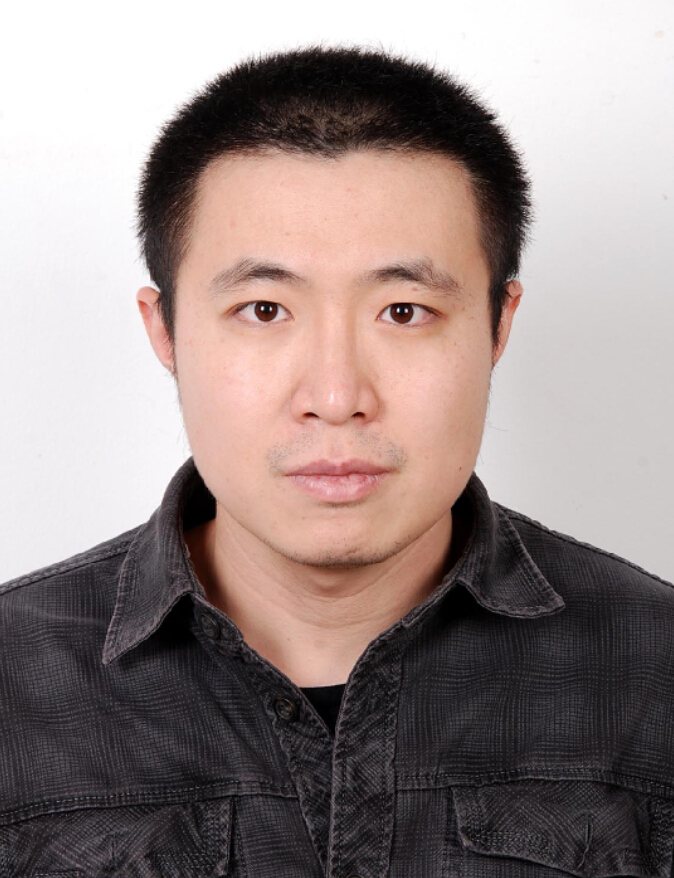}}]{Xiao Lu} (Member, IEEE) received the B.Eng. degree in communication engineering from the Beijing University of Posts and Telecommunications, Beijing, China, the M.Eng. degree in computer engineering from Nanyang Technological University, Singapore, and the Ph.D. degree from the University of Alberta, Edmonton, AB, Canada. He is currently with Ericsson, Canada and York University, Canada. His current research interests include the design, analysis, and optimization of future generation cellular wireless networks. He was a recipient of the Best Paper Award from the IEEE Communication Society: Green Communications ${\rm{\& }}$ Computing Technical Committee in 2018 and the IEEE TCGCC Best Conference Paper.
\end{IEEEbiography}

\begin{IEEEbiography}
[{\includegraphics[width=1in,height=1.25in,clip,keepaspectratio]{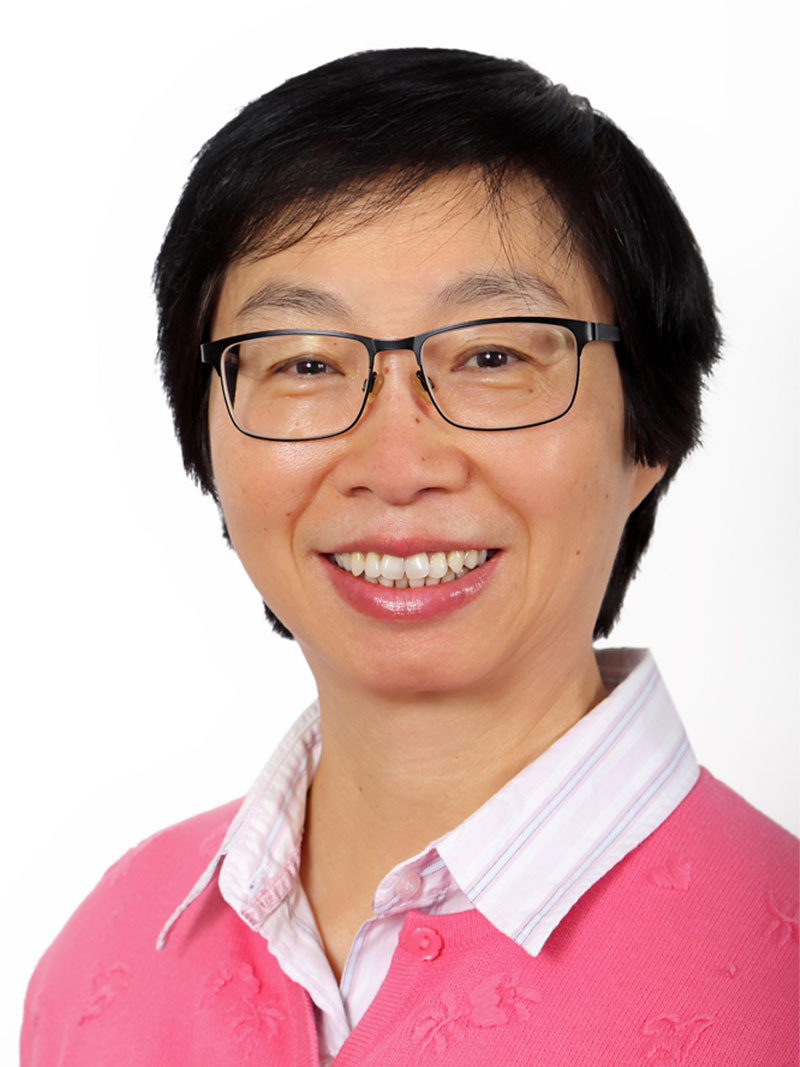}}]{Weihua Zhuang} (Fellow, IEEE) received the B.Sc. and M.Sc. degrees from Dalian Maritime University, China, and the Ph.D. degree from the University of New Brunswick, Canada, all in electrical engineering. Since 1993, she has been a faculty member in the Department of Electrical and Computer Engineering, University of Waterloo, Canada, where she is a University Professor and a Tier I Canada Research Chair in wireless communication networks. Her current research focuses on network architecture, algorithms and protocols, and service provisioning in future communication systems. She is the recipient of Women’s Distinguished Career Award in 2021 from IEEE Vehicular Technology Society, R.A. Fessenden Award in 2021 from IEEE Canada, Award of Merit in 2021 from the Federation of Chinese Canadian Professionals (Ontario), and Technical Recognition Award in Ad Hoc and Sensor Networks in 2017 from IEEE Communications Society. Dr. Zhuang is a Fellow of the IEEE, Royal Society of Canada (RSC), Canadian Academy of Engineering (CAE), and Engineering Institute of Canada (EIC). She is the President and an elected member of the Board of Governors (BoG) of the IEEE Vehicular Technology Society. She was the Editor-in-Chief of IEEE Transactions on Vehicular Technology (2007-2013), an editor of IEEE Transactions on Wireless Communications (2005-2009), General Co-Chair of IEEE/CIC International Conference on Communications in China (ICCC) 2021, Technical Program Committee (TPC) Chair/Co-Chair of IEEE Vehicular Technology Conference 2017 Fall and 2016 Fall, TPC Symposia Chair of the IEEE Globecom 2011, and an IEEE Communications Society Distinguished Lecturer (2008-2011).
\end{IEEEbiography}

\begin{IEEEbiography}
[{\includegraphics[width=1in,height=1.25in,clip,keepaspectratio]{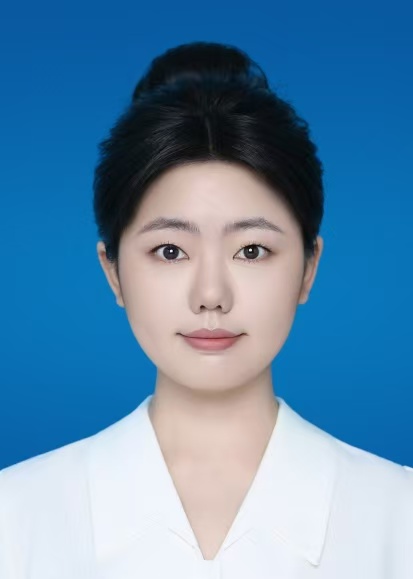}}]{Shuang Qi} Shuang Qi received the B.S. and M.S. degrees from Liaoning Normal University, Dalian, China, in 2017 and 2020, respectively and the Ph.D. degree from the College of Information Science and Technology, Dalian Maritime University, Dalian, China, in 2025. She is currently a lecturer at the School of Electronic and Information Engineering, Liaoning University of Technology. Her research interests include maritime communication networks, mobile edge computing, and resource management.
\end{IEEEbiography}
\vfill

\end{document}